\documentclass[11pt]{article}

\usepackage{amsmath} % AMS Math Package
\usepackage{amsthm} % Theorem Formatting
\usepackage{amssymb}	% Math symbols such as \mathbb
\usepackage{graphicx} % Allows for eps images
\usepackage{multicol} % Allows for multiple columns
\usepackage{color}
\usepackage[dvips,letterpaper,margin=1in,bottom=1in]{geometry}

\usepackage[utf8]{inputenc}
\usepackage[english]{babel}

\newtheorem{theorem}{Theorem}[section]

\newtheorem{lemma}[theorem]{Lemma}

\newtheorem{proposition}[theorem]{Proposition}
\newtheorem{definition}{Definition}[section]
\newtheorem{example}{Example}

\newtheorem{problem}{Problem}

\newcommand{\abs}[1]{\left| #1 \right|} % for absolute value
\newcommand{\ket}[1]{\left| #1 \right>} % for Dirac bras
\newcommand{\bra}[1]{\left< #1 \right|} % for Dirac kets
\usepackage{latexsym}
\usepackage{CJK}

\usepackage{qcircuit}

\usepackage{enumerate}
\usepackage{ulem}

\usepackage{algorithm}
\usepackage{algpseudocode}
\usepackage{stmaryrd}
\usepackage{booktabs}

\usepackage{hyperref}
\newcommand{\footremember}[2]{%
    \footnote{#2}
    \newcounter{#1}
    \setcounter{#1}{\value{footnote}}%
}

\begin{document}
	%\begin{CJK*}{GBK}{song}

    \title{Equivalence Checking of Quantum Finite-State Machines}
        \author{
            Qisheng Wang \footremember{one}{Qisheng Wang is with the Department of Computer Science and Technology, Tsinghua University, China (e-mail: \url{QishengWang1994@gmail.com}).}
            \and Junyi Liu \footremember{two}{Junyi Liu is with the State Key Laboratory of Computer Science, Institute of Software, Chinese Academy of Sciences, China, and also with the University of Chinese Academy of Sciences, China (e-mail: \url{liujy@ios.ac.cn}).}
            \and Mingsheng Ying \footremember{three}{Mingsheng Ying is with Centre for Quantum Software and Information, University of Technology Sydney, Australia, the State Key Laboratory of Computer Science, Institute of Software, Chinese Academy of Sciences, China, and also with the Department of Computer Science and Technology, Tsinghua University, China (e-mail: \url{yingms@ios.ac.cn}).}
        }
        \date{}
        \maketitle

    \begin{abstract}
    In this paper, we introduce the model of quantum Mealy machines and study the equivalence checking and minimisation problems of them. Two efficient algorithms are developed for checking equivalence of two states in the same machine and for checking equivalence of two machines. As an application, they are used in equivalence checking of quantum circuits. Moreover, the minimisation problem is proved to be in \textbf{PSPACE}.
    \end{abstract}

    \textbf{Keywords: quantum computing, quantum circuits, Mealy machines, equivalence checking, minimisation.}

    \newpage

    \tableofcontents
    \newpage

\section{Introduction}

 A large variety of real-world testing, analysis and verification problems for computer and communication hardware and software can be reduced to equivalence checking of Mealy machines \cite{Jha}, \cite{Lee94}, \cite{Lee96}. The same problem has emerged in the quantum realm with the rapid progress of quantum information technology in recent years; for example, equivalence checking of quantum circuits \cite{Via07}, \cite{Markov} and quantum communication protocols \cite{Eb13}, \cite{Eb14}, property testing  \cite{Sei12}, fault detection and diagnosis \cite{Ban10}, \cite{Ber18}, \cite{Bia10}, \cite{Pal12}, reachability analysis \cite{Hung04} and test generation \cite{Per05} of quantum circuits. But up to now, they are investigated separately in ad hoc manners without a unified model.

The \textit{overall aim} of this paper is to introduce a quantum generalisation of Mealy machines with the hope that our results can provide a formal model and some useful theoretical tools for solving these problems.
As determined by the basic postulates of quantum mechanics, the state space of a quantum Mealy machine is a (finite-dimensional) Hilbert space, its dynamics is modelled by unitary operators, and its outputs come as the outcomes of certain quantum measurements.

This paper studies two central problems, namely equivalence checking and minimisation, of quantum Mealy machines. As in classical Mealy machines, equivalence checking is carried out by inputting a sequence into the checked machines and then observing their respective outputs. A major difference between the classical and quantum cases is caused by the fact that quantum measurements can change the states of the observed systems. Consequently, a notion of scheduler must be introduced in the quantum case to specify the locations where quantum measurements are designed to perform.

{\vskip 4pt}

\textbf{Main Technical Contributions} include:
 \begin{itemize}\item We develop two algorithms for equivalence checking of complexity $O(mn^6)$, where $m$ is the number of input and output symbols and $n$ the dimension of the state Hilbert spaces of the checked machines.
\item The minimisation problem is proved to be in \textbf{PSPACE}.
 \end{itemize}

As an application, our algorithms are used for checking equivalence of quantum circuits in 30 benchmarks.

Quantum generalisations of various automata have been extensively studied in the literature; see for example \cite{Kon97}, \cite{Moo00}.
The problems of equivalence checking and minimisations for quantum automata rather that quantum Mealy machines defined in this paper have already been considered in a series of papers  \cite{Kos01}, \cite{Li06}, \cite{Li062}, \cite{Li08}, \cite{Li12}, \cite{Qiu11}, \cite{Tze92}, \cite{Wan18}. The techniques developed in this paper can be used to improve some of their complexity results.

 {\vskip 4pt}

\textbf{Organisation of the Paper}: The notion of quantum Mealy machine is defined and its behaviour is described in Sec. \ref{defin}. Our main results including two algorithms are given in Sec. \ref{main}. The improvements over the complexity results for other quantum automata with our new techniques are also briefly discussed there. The case studies for equivalence checking of benchmark quantum circuits are described in Sec. \ref{test}. The proofs of our main theorems are presented in Sec. \ref{proofs}.
For readability, the proofs of other results are deferred into the Appendices. A short conclusion is drawn in Sec. \ref{conclusion}.

    \section{Basic Definitions}\label{defin}

Let us first very briefly review several basic notions in quantum mechanics. The state space of a quantum system is a Hilbert space. For an integer $n\geq 1$, an $n$-dimensional Hilbert space $\mathcal{H}$ is essentially the space $\mathbb{C}^n$ of $n$-dimensional vectors of complex numbers with the ordinary inner product. Using Dirac's notation, a vector in $\mathcal{H}$ is denoted $|\psi\rangle$, and the inner product of $|\psi\rangle$ and $|\varphi\rangle$ is written $\langle\psi|\varphi\rangle$. A pure state of the quantum system is then described by a vector $|\psi\rangle$ of length $$\|\psi\|=\sqrt{\langle\psi|\psi\rangle}=1.$$ For example, a qubit lives in $\mathbb{C}^2$ and it can be in a basis state $$|0\rangle=\begin{bmatrix} 1 \\ 0 \end{bmatrix}\ {\rm or}\ |1\rangle=\begin{bmatrix} 0 \\ 1 \end{bmatrix},$$ or a superposition of them like $$|\pm\rangle=\frac{1}{\sqrt{2}}(|0\rangle\pm|1\rangle)=\frac{1}{\sqrt{2}}\begin{bmatrix} 1 \\ \pm 1 \end{bmatrix}.$$
An operator in $\mathcal{H}$ is represented by an $n\times n$ matrix $A=\left[A_{ij}\right]$. The trace of $A$ is defined as $$\operatorname{tr}(A)=\sum_iA_{ii}.$$ Then a mixed state of the quantum system is expressed by a density operator, i.e. a positive semidefinite matrix $\rho$ with $\operatorname{tr}(\rho)=1$. Furthermore, an action on the system causes a certain evolution:
$$|\psi\rangle\rightarrow U|\psi\rangle\ {\rm (pure\ state)\ or}\ \rho\rightarrow U\rho U^\dag\ {\rm (mixed\ state)}$$
 modelled by a unitary operator, i.e. a matrix $U$ with $U^\dag U=I$, where $U^\dag$ stands for the complex conjugate transpose of $U$, and $I$ the unit matrix. For example, Hadamard gate $$H=\frac{1}{\sqrt{2}}\begin{bmatrix} 1 & 1\\ 1 &-1 \end{bmatrix}$$  transforms $|0\rangle$ to $|+\rangle$ and $|1\rangle$ to $|-\rangle$. A quantum measurement is used to readout the outcomes in quantum computing. Mathematically, it is described by a set of operators $M=\left\{M_m\right\}$ with the normalisation condition $$\sum_m M_m^\dag M_m=I.$$ If we perform it on quantum system in pure state $|\psi\rangle$, then outcome $m$ is obtained with probability $$p_m=\|M_m|\psi\rangle\|^2,$$ and after that the system is in state $$M_m|\psi\rangle/\sqrt{p_m};$$ and if we perform it on mixed state $\rho$, then outcome $m$ is obtained with probability $$p_m = \operatorname{tr}(M_m \rho M_m^\dag)$$ and the system collapses to $$M_m \rho M_m^\dag /p_m.$$ For example, if we measure qubit $|+\rangle$ in the computation basis, i.e. the measurement is $$M = \left\{M_0 =
|0\rangle\langle 0|, M_1 = |1\rangle\langle 1|\right\},$$ then outcomes $0$ and $1$ are observed with equal probability $\frac{1}{2}$, and after that the qubit is in state $|0\rangle$ or $|1\rangle$, respectively.

The quantum generalisations of various computational models (e.g. finite-state automata, pushdown automata and Turing machines) have been defined in the literature by incorporating the above quantum mechanical ideas into these models. Similarly, combining these ideas with the classical Mealy machine model \cite{Mea55} yields straightforwardly:

    \begin{definition}[Quantum Mealy Machine]
    A quantum Mealy machine (QMM for short) is a $5$-tuple
    $\mathcal{M} = (\Sigma, \Gamma, \mathcal{H}, U, M),$
    where:
    \begin{enumerate}
      \item[-] $\Sigma$ is a finite input alphabet;
      \item[-] $\Gamma$ is a finite output alphabet;
      \item[-] $\mathcal{H}$ is a finite-dimensional Hilbert space;
      \item[-] $U = \{ U_\sigma: \sigma \in \Sigma \}$ is a set of unitary operators. For each $\sigma \in \Sigma$, $U_\sigma$ is a unitary operator on $\mathcal{H}$; and
      \item[-] $M = \{ M_m : m \in \Gamma \}$ is a quantum measurement in $\mathcal{H}$, that is, $M_m$ is a linear operator on $\mathcal{H}$ for each $m \in \Gamma$ and $\sum_m M_m^\dag M_m = I$.
    \end{enumerate}
    \end{definition}

Similar to the case of other computational models and their quantum counterparts, there is a major difference between classical and quantum Mealy machines. As is well-known, in order to extract information about a quantum system, we have to perform a measurement on it. On the other hand, a measurement can change the state of the system. So, the dynamic behaviour of the system depends heavily on the time points where the measurement is performed. This motivates us to introduce the notion of (measurement) scheduler.
 For a finite string (word) $a \in \Sigma^*$ on an alphabet $\Sigma$, let $\abs{a}$ stand for the length of $a$, $a[i]$ be the $i$-th symbol of $a$ ($1$-indexed), and $a[l:r]$ denote the substring $a[l]a[l+1]\dots a[r]$ of $a$. Especially, in the case of $l > r$, $a[l:r]$ is the empty string $\epsilon$.

    \begin{definition}[Scheduler]
        Let $a \in \Sigma^*$ be an input word. A scheduler for $a$ is a non-decreasing sequence $\mathcal{S} = \{ s_i \}$ with $0 \leq s_1 \leq s_2 \leq \dots \leq s_{\abs{\mathcal{S}}} \leq \abs{a}$.
         The set of schedulers for $a$ is denoted $\mathfrak{S}_a$. Moreover, the set of all schedulers is denoted $$\mathfrak{S} = \bigcup_{a \in \Sigma^*} \mathfrak{S}_a.$$
    \end{definition}

    Intuitively, each $s_i$ represents a location where a measurement is scheduled to perform. If $s_{\abs{\mathcal{S}}} = \abs{a}$, that is, a measurement is performed at the end of $a$, then $\mathcal{S}$ is called \textit{closed}.

    %Let $\mathfrak{S}_a$ denote the set of schedulers for $a$ and $\mathfrak{S} = \bigcup_{a \in \Sigma^*} \mathfrak{S}_a$ be the set of all schedulers.

    Let us see how a QMM $\mathcal{M}$ runs. For any word $a\in\Sigma^*$, we write $$U_a = U_{a[\abs{a}]} \dots U_{a[2]} U_{a[1]}$$ (the composition of unitary transformations, or equivalently the multiplication of unitary matrices); in particular, $U_\epsilon = I$ for the empty word. For a Hilbert space $\mathcal{H}$, let $\mathcal{D}(\mathcal{H})$ be the set of density operators on $\mathcal{H}$.
 Suppose that the initial state is $\rho\in \mathcal{D}(\mathcal{H})$, the input word is $a \in \Sigma^*$ and $\mathcal{S} = \{ s_i \}$ is a scheduler for $a$. The scheduler $\mathcal{S}$ splits $a$ into $(\abs{\mathcal{S}}+1)$ parts: $a_i = a[s_{i-1}+1:s_i]$ for $1 \leq i \leq \abs{\mathcal{S}}+1$, where $s_0 = 0$ and $s_{\abs{\mathcal{S}}+1} = \abs{a}$. Machine $\mathcal{M}$ performs measurement $M$ exactly $\abs{\mathcal{S}}$ times according to $\mathcal{S}$: starting from $\rho$, for each $1 \leq i \leq \abs{\mathcal{S}}$, $\mathcal{M}$ first applies unitary $U_{a_i}$ on the system, then performs measurement $M$ and produces an outcome $b_i \in \Gamma$. Thus, an output word $b= b_1 b_2 \dots b_{\abs{\mathcal{S}}}\in \Gamma^*$ is produced with probability:    \[
        \Pr\nolimits_{\rho}^{\mathcal{M}} (b|a, \mathcal{S}) = \operatorname{tr} \left( \rho_{b|a, \mathcal{S}}^{\mathcal{M}} \right),
    \]
    where $$\rho_{b|a, \mathcal{S}}^{\mathcal{M}} = V_{b|a, \mathcal{S}} \rho V_{b|a, \mathcal{S}}^\dag\ {\rm and}\ V_{b|a, \mathcal{S}} = U_{a_{\abs{\mathcal{S}}+1}} M_{b_{\abs{\mathcal{S}}}} U_{a_{\abs{\mathcal{S}}}} \dots M_{b_1} U_{a_1}.$$ The final state of $\mathcal{M}$ is    \[
        \rho' = \frac {\rho_{b|a, \mathcal{S}}^{\mathcal{M}}} {\Pr_{\rho}^{\mathcal{M}} (b|a, \mathcal{S})}.
    \]

    Now we are ready to define the central notion of this paper: equivalence of two states in a quantum Mealy machine.

    \begin{definition}[Equivalence of States]
        Given a QMM $\mathcal{M}$ and two states $\rho_s$ and $\rho_t$.
        \begin{enumerate}
          \item $\rho_s$ and $\rho_t$ are equivalent, denoted $\rho_s \sim \rho_t$, if for every input $a \in \Sigma^*$ and scheduler $\mathcal{S}$ and output $b \in \Gamma^{\abs{\mathcal{S}}}$,
              \begin{equation} \label{eq1}
                \Pr\nolimits_{\rho_s}^{\mathcal{M}} (b|a, \mathcal{S}) = \Pr\nolimits_{\rho_t}^{\mathcal{M}} (b|a, \mathcal{S}).
              \end{equation}
          \item $\rho_s$ and $\rho_t$ are equivalent up to $k$ measurements, denoted $\rho_s \sim_k \rho_t$, if Eq. (\ref{eq1}) holds for all schedulers $\mathcal{S}$ with $\abs{\mathcal{S}} \leq k$.
          \item $\rho_s$ and $\rho_t$ are $m$-equivalent, denoted $\rho_s \sim^m \rho_t$, if Eq. (\ref{eq1}) holds for all inputs $a \in \Sigma^*$ and schedulers $\mathcal{S}$ with $\abs{a}+\abs{\mathcal{S}} \leq m$.

          \item $\rho_s$ and $\rho_t$ are $m$-equivalent up to $k$ measurements, denoted $\rho_s \sim_{k}^m \rho_t$, if Eq. (\ref{eq1}) holds for all inputs $a \in \Sigma^*$ and schedulers $\mathcal{S}$ with $\abs{\mathcal{S}} \leq k$ and $\abs{a}+\abs{\mathcal{S}} \leq m$.
        \end{enumerate}
    \end{definition}

    An input word $a$ together with a scheduler $\mathcal{S}$ for it is called an experiment, and $|a|+|\mathcal{S}|$ is called the size of the experiment.
    The notion of equivalence in the above definition was introduced only for two states in the same quantum Mealy machines. But it can be simply generalised to compare two states in different machines.

    \begin{definition}[Equivalence of Machines] Given two QMMs $\mathcal{M}_1$ and $\mathcal{M}_2$ with the same input alphabet $\Sigma$ and output alphabet $\Gamma$, and their initial states $\rho_1, \rho_2$.
        \begin{enumerate}
          \item $(\mathcal{M}_1, \rho_1)$ and $(\mathcal{M}_2, \rho_2)$ are equivalent, denoted $(\mathcal{M}_1, \rho_1) \sim (\mathcal{M}_2, \rho_2)$, if for every $a \in \Sigma^*$ and scheduler $\mathcal{S}$ and output $b \in \Gamma^{\abs{\mathcal{S}}}$,
              \begin{equation} \label{eq2}
                \Pr\nolimits_{\rho_1}^{\mathcal{M}_1} (b|a, \mathcal{S}) = \Pr\nolimits_{\rho_2}^{\mathcal{M}_2} (b|a, \mathcal{S}).
              \end{equation}
          \item $(\mathcal{M}_1, \rho_1)$ and $(\mathcal{M}_2, \rho_2)$ are equivalent up to $k$ measurements, denoted $(\mathcal{M}_1, \rho_1) \sim_k (\mathcal{M}_2, \rho_2)$, if Eq. (\ref{eq2}) holds for all schedulers $\mathcal{S}$ with $\abs{\mathcal{S}} \leq k$.
          \item $(\mathcal{M}_1, \rho_1)$ and $(\mathcal{M}_2, \rho_2)$ are $m$-equivalent, denoted $(\mathcal{M}_1, \rho_1) \sim^m (\mathcal{M}_2, \rho_2)$, if Eq. (\ref{eq2}) holds for all inputs $a \in \Sigma^*$ and schedulers $\mathcal{S}$ with $\abs{a}+\abs{\mathcal{S}} \leq m$.
          \item $(\mathcal{M}_1, \rho_1)$ and $(\mathcal{M}_2, \rho_2)$ are $m$-equivalent up to $k$ measurements, denoted $(\mathcal{M}_1, \rho_1) \sim_k^m (\mathcal{M}_2, \rho_2)$, if Eq. (\ref{eq2}) holds for all inputs $a \in \Sigma^*$ and schedulers $\mathcal{S}$ with $\abs{\mathcal{S}} \leq k$ and $\abs{a}+\abs{\mathcal{S}} \leq m$.
        \end{enumerate}
    \end{definition}

    We present two examples to illustrate how the notions defined above can be used to model quantum circuits and their equivalence.

    \begin{example} [Quantum Circuits under Resource Constraints] \label{example1}
        In real world, there are usually certain restrictions on the gates in a quantum circuit.
        Let us consider a quantum circuit with two qubits $x_1$ and $x_2$. Suppose only two kinds of quantum gates are available, which are the Hadamard gate on the first qubit, denoted $U_{H_1} = H[x_1]$, and the CNOT gate with $x_1$ as its control qubit, denoted  $U_C = \mathit{CNOT}[x_1, x_2]$.
        Also suppose we can only measure the first qubit in the computational basis. The measurement can be described as $M=\{M_0,M_1\}$ with $$M_0 = \ket{00}\bra{00}+\ket{01}\bra{01},\qquad  M_1 = \ket{10}\bra{10}+\ket{11}\bra{11}.$$ This kind of quantum circuits can be modelled by a QMM
        \[
        \mathcal{M} = ( \Sigma, \Gamma, \mathcal{H}_2^{\otimes 2}, U, M ),
        \]
   where $\Sigma = \{ C, {H_1} \}$, $\Gamma = \{ 0, 1 \}$ and $U = \{ U_C, U_{H_1} \}$.
        Consider two states $\ket{00}$ and $\ket{01}$, and our question is: can we distinguish them using such a quantum circuit? The answer is \textquotedblleft\textit{no}\textquotedblright\ because $\ket{00} \sim \ket{01}$ in $\mathcal{M}$.

        Now we loosen the restriction and allow the Hadamard gate to act on the second qubit, denoted $U_{H_2} = H[x_2]$. This kind of quantum circuits can be described by a QMM
        \[
        \mathcal{M}' = ( \Sigma^\prime, \Gamma, \mathcal{H}_2^{\otimes 2}, U^\prime, M ),
        \]
where $\Sigma^\prime = \{ C, H_1, H_2 \}$ and $U^\prime = \{ U_C, U_{H_1}, U_{H_2} \}.$ It can be verified directly that $\ket{00} \sim^{3} \ket{01}$ in $\mathcal{M}'$. However, $\ket{00} \not \sim \ket{01}$ in $\mathcal{M}'$ because
        \[
        \Pr\nolimits_{\ket{00}\bra{00}}^{\mathcal{M}'}(0|H_1 H_2 C H_1, \{ 4 \}) = 1,\qquad
        \Pr\nolimits_{\ket{01}\bra{01}}^{\mathcal{M}'}(0|H_1 H_2 C H_1, \{ 4 \}) = 0.
        \]
        This means that states $\ket{00}$ and $\ket{01}$ can be distinguished in $\mathcal{M}^\prime$ by experiments of size 5 but not by those of size only 3.
        A quantum circuit distinguishing $\ket{00}$ and $\ket{01}$ is given in Fig. \ref{fig1}.
        \begin{figure}\centering
        \[
        \Qcircuit @C=1em @R=.7em {
            & x_1 & & \qw & \gate{H} & \ctrl{1} & \gate{H} & \meter  \\
            & x_2 & & \qw & \gate{H} & \targ & \qw & \qw  \\
        }
        \]
        \caption{A quantum circuit that distinguishes $\ket{00}$ and $\ket{01}$ using gates $H[x_1]$, $H[x_2]$, $\mathit{CNOT}[x_1, x_2]$ and measurement $M[x_1]$.}
        \label{fig1}
        \end{figure}
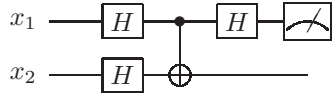\end{example}

It is well-known \cite{Paz71} that two states in an $n$-dimensional probabilistic Mealy machine are equivalent, if and only if they are $(n-1)$-equivalent. The above example shows an interesting difference between quantum and probabilistic Mealy machines: in the $4$-dimensional QMM $\mathcal{M}'$, $\ket{00} \sim^3 \ket{01}$ but $\ket{00} \not \sim \ket{01}$; more precisely, $\ket{00} \not \sim_1 \ket{01}$.

    \begin{example} [Quantum Circuits with Multi-Measurements] \label{example2}
        Let's consider again a quantum circuit with two qubits $x_1$ and $x_2$.
        But we suppose the two available quantum gates are Hadamard gate on the first qubit $U_{H} = H[x_1]$ and the swap gate on $x_1$ and $x_2$:
        \[
        U_S = \begin{bmatrix}
            1 & 0 & 0 & 0 \\
            0 & 0 & 1 & 0 \\
            0 & 1 & 0 & 0 \\
            0 & 0 & 0 & 1
        \end{bmatrix},
        \]
        i.e. $U_S \ket {i, j} = \ket{j, i}$ for $i, j \in \{0, 1\}$.
        Moreover, we can only measure the first qubit in the computational basis. This kind of circuits can be described by a QMM
        \[
        \mathcal{M} = (\Sigma, \Gamma, \mathcal{H}_2^{\otimes 2}, U, M),
        \] where $\Sigma = \{ H, S \}$, $U = \{ U_H, U_S \}$ and others are the same as in Example \ref{example1}.
        Now we consider two (entangled) Bell states $\ket{\beta_{00}}$ and $\ket{\beta_{10}}$, where
        \[
            \beta_{xy} = \frac 1 {\sqrt2} (\ket{0y}+(-1)^x\ket{1\bar y})
        \]
         and $\bar y$ is the negation of $y$.
        It is easy to verify that $\ket{\beta_{00}} \sim_1 \ket{\beta_{10}}$, which means that we cannot distinguish $\ket{\beta_{00}}$ and $\ket{\beta_{10}}$ using only one measurement. However, $\ket{\beta_{00}} \not \sim_2 \ket{\beta_{10}}$. Indeed, they are distinguished by input word $a = HSH$ and scheduler $\mathcal{S} = \{ 1, 3 \}$; the corresponding quantum circuit is given in Fig. \ref{fig2}.
        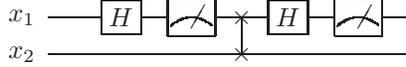
\begin{figure}\centering
        \[
        \Qcircuit @C=1em @R=.7em {
            & x_1 & & \qw & \gate{H} & \meter & \qswap & \gate{H} & \meter & \qw \\
            & x_2 & & \qw & \qw & \qw & \qswap \qwx & \qw & \qw & \qw \\
        }
        \]
        \caption{A quantum circuit that distinguishes $\ket{\beta_{00}}$ and $\ket{\beta_{10}}$ with gates $H[x_1]$, $S[x_1, x_2]$ and measurement $M[x_1]$.}
        \label{fig2}
        \end{figure}
    \end{example}

    \section{Main Results}\label{main}

    In this section, we present the main results of this paper; for readability, some of their proofs are postponed to Sec. \ref{proofs}, and some are further deferred into the Appendices.

     \subsection{Checking equivalence of states}

     First, we consider equivalence checking of two states in the same QMM. The following theorem establishes an upper bound for the size of experiments required for equivalence checking in terms of the dimension of the state Hilbert space.

    \begin{theorem} \label{thm-state-eq}
        Given a QMM $\mathcal{M}$ with state Hilbert space $\mathcal{H}$ and two states $\rho_s$ and $\rho_t$. Let $n = \dim \mathcal{H}$. Then:
        \begin{enumerate}
          \item $\rho_s \sim\rho_t \Longleftrightarrow \rho_s \sim^{n^2-1} \rho_t \Longleftrightarrow \rho_s \sim_{n^2-1} \rho_t$.
          \item For every $k \in \mathbb{N}$, $\rho_s \sim_{k} \rho_t \Longleftrightarrow \rho_s \sim_{k}^{n^2-1} \rho_t$. \end{enumerate}
    \end{theorem}

    If $\mathcal{M}$ is \textit{real}; that is, all of its unitary matrices $U_\sigma$ and measurement matrices $M_m$ consist of real entries, then the experiment size $n^2-1$ can be improved to $\frac 1 2 n(n+1)-1$.

    An algorithm for checking equivalence of states in a QMM can be directly derived from Theorem \ref{thm-state-eq} by enumerating all possible inputs $a$ and schedulers $\mathcal{S}$ with $\abs{a}+\abs{\mathcal{S}} \leq n^2-1$, but its complexity is $(\abs{\Sigma}+\abs{\Gamma})^{O(n^2)},$ exponential in $n$, the dimension of the state Hilbert space. We are able to develop a much more efficient algorithm with a time complexity polynomial in $n$. Let $\Sigma = \{ \sigma_1, \sigma_2, \dots, \sigma_{\abs{\Sigma}} \}$ and $\Gamma = \{ \gamma_1, \gamma_2, \dots, \gamma_{\abs{\Gamma}} \}$.
    The algorithms for the cases without and with a bound on the number of measurements are described in Algorithm \ref{algo1} and Algorithm \ref{algo2}, respectively. Their complexities are given in the next theorem.
    \iffalse
    The algorithm is described as follows:
    \begin{enumerate}
      \item[(1)] Initially, set the set $\rho = \rho_s - \rho_t$, $\mathfrak{B} = \emptyset$ and push $(\epsilon, \emptyset, \epsilon)$ into the empty queue $Q$.
      \item[(2)] If $Q$ is empty, goto (3). Otherwise, pop the front element $(a, \mathcal{S}, b)$ of $Q$. If $\rho_{b|(a, \mathcal{S})} \notin \operatorname{span} \mathfrak{B}$, then
          \begin{enumerate}
            \item Add $\rho_{b|(a, \mathcal{S})}$ into $\mathfrak{B}$.
            \item Push $(a\sigma_i, \mathcal{S}, b)$ into $Q$ for $1 \leq i \leq \abs{\Sigma}$ in turn.
            \item Push $(a, \mathcal{S} \cup \{ \abs{a} \}, b\gamma_i)$ into $Q$ for $1 \leq i \leq \abs{\Gamma}$ in turn.
          \end{enumerate}
          Goto (2).
      \item [(3)] Return true if $\operatorname{tr}(\varrho) = 0$ for every $\varrho \in \mathfrak{B}$, and false otherwise.
    \end{enumerate}
    See Algorithm 1 for the pseudo-code of the algorithm.
    \fi

    \begin{algorithm}
        \caption{A polynomial algorithm for checking whether $\rho_s \sim \rho_t$ in $\mathcal{M}$.}
        \label{algo1}
        \begin{algorithmic}[1]
        \Require QMM $\mathcal{M} = \{ \Sigma, \Gamma, \mathcal{H}, U, M \}$ and two density operators $\rho_s$ and $\rho_t$.
        \Ensure Whether $\rho_s \sim \rho_t$ or not.

        \State $\rho \gets \rho_s - \rho_t$.
        \State $\mathfrak{B} \gets \emptyset$.
        \State Let $Q$ be an empty queue and push $(\epsilon, \emptyset, \epsilon)$ into $Q$.
        \While {$Q$ is not empty}
            \State Pop the front element $(a, \mathcal{S}, b)$ of $Q$.
            \If {$\rho_{b|(a, \mathcal{S})} \notin \operatorname{span} \mathfrak{B}$}
                \State Add $\rho_{b|(a, \mathcal{S})}$ into $\mathfrak{B}$.
                \State Push $(a\sigma_i, \mathcal{S}, b)$ into $Q$ for $1 \leq i \leq \abs{\Sigma}$ in turn.
                \State Push $(a, \mathcal{S} + \{ \abs{a} \}, b\gamma_i)$ into $Q$ for $1 \leq i \leq \abs{\Gamma}$ in turn.
            \EndIf
        \EndWhile

        \If {$\operatorname{tr}(\varrho) = 0$ for every $\varrho \in \mathfrak{B}$}
            \State \Return \textbf{true}.
        \Else
            \State Find an arbitrary $\rho_{b|(a, \mathcal{S})} \in \mathfrak{B}$ such that $\operatorname{tr}\left(\rho_{b|(a, \mathcal{S})}\right) \neq 0$.
            \State \Return \textbf{false} with witness $(a, \mathcal{S}, b)$.
        \EndIf

        \end{algorithmic}
    \end{algorithm}

    \begin{algorithm}
        \caption{A polynomial algorithm for checking whether $\rho_s \sim_k \rho_t$ in $\mathcal{M}$.}
        \label{algo2}
        \begin{algorithmic}[1]
        \Require QMM $\mathcal{M} = \{ \Sigma, \Gamma, \mathcal{H}, U, M \}$, integer $k$ and two density operators $\rho_s$ and $\rho_t$.
        \Ensure Whether $\rho_s \sim_k \rho_t$ or not.

        \State $\rho \gets \rho_s - \rho_t$.
        \State $\mathfrak{B}_i \gets \emptyset$ for $0 \leq i \leq k$.
        \State Let $Q$ be an empty queue and push $(\epsilon, \emptyset, \epsilon)$ into $Q$.
        \While {$Q$ is not empty}
            \State Pop the front element $(a, \mathcal{S}, b)$ of $Q$.
            \If {$\rho_{b|(a, \mathcal{S})} \notin \operatorname{span} \mathfrak{B}_{\abs{\mathcal{S}}}$}
                \State Find the largest $\abs{\mathcal{S}} \leq j \leq k$ such that $\rho_{b|(a, \mathcal{S})} \notin \operatorname{span} \mathfrak{B}_j$.
                \State Add $\rho_{b|(a, \mathcal{S})}$ into $\mathfrak{B}_l$ for $\abs{\mathcal{S}} \leq l \leq j$.
                \State Push $(a\sigma_i, \mathcal{S}, b)$ into $Q$ for $1 \leq i \leq \abs{\Sigma}$ in turn.
                \If {$\abs{\mathcal{S}} < k$}
                    \State Push $(a, \mathcal{S} + \{ \abs{a} \}, b\gamma_i)$ into $Q$ for $1 \leq i \leq \abs{\Gamma}$ in turn.
                \EndIf
            \EndIf
        \EndWhile

        \If {$\operatorname{tr}(\varrho) = 0$ for every $\varrho \in \mathfrak{B}_k$}
            \State \Return \textbf{true}.
        \Else
            \State Find an arbitrary $\rho_{b|(a, \mathcal{S})} \in \mathfrak{B}_k$ such that $\operatorname{tr}\left(\rho_{b|(a, \mathcal{S})}\right) \neq 0$.
            \State \Return \textbf{false} with witness $(a, \mathcal{S}, b)$.
        \EndIf

        \end{algorithmic}
    \end{algorithm}

    \begin{theorem} \label{thm-algo}
        Given a QMM $\mathcal{M} = (\Sigma, \Gamma, \mathcal{H}, U, M)$, two states $\rho_s$ and $\rho_t$ and a positive integer $k$. Let $m = \abs{\Sigma}+\abs{\Gamma}$ and $n = \dim \mathcal{H}$. Then:
        \begin{enumerate}
          \item There is an $O(mn^6)$ algorithm that decides whether $\rho_s \sim\rho_t$; if not, it finds an input $a \in \Sigma^*$ and a closed scheduler $\mathcal{S}$ with $\abs{a}+\abs{\mathcal{S}} \leq n^2-1$ and output $b \in \Gamma^{\abs{\mathcal{S}}}$ such that $$\Pr\nolimits_{\rho_s}^{\mathcal{M}}(b|a, \mathcal{S}) \neq \Pr\nolimits_{\rho_t}^{\mathcal{M}}(b|a, \mathcal{S}).$$
          \item There is an $O(kmn^6)$ algorithm that decides whether $\rho_s \sim_{k} \rho_t$; if not, it finds an input $a \in \Sigma^*$ and a closed scheduler $\mathcal{S}$ with $\abs{a}+\abs{\mathcal{S}} \leq n^2-1$ and $\abs{\mathcal{S}} \leq k$ and output $b \in \Gamma^{\abs{\mathcal{S}}}$ such that $$\Pr\nolimits_{\rho_s}^{\mathcal{M}}(b|a, \mathcal{S}) \neq \Pr\nolimits_{\rho_t}^{\mathcal{M}}(b|a, \mathcal{S}).$$
        \end{enumerate}
    \end{theorem}

     \subsection{Checking equivalence of machines}

    Now we turn to consider equivalence checking of two QMMs. The basic idea is to reduce this problem to the problem examined in the previous subsection. For two Hilbert spaces $\mathcal{H}_1$ and $\mathcal{H}_2$, $\mathcal{H}_1\oplus\mathcal{H}_2$ stands for their direct sum. If $A_1,A_2$ are two matrices (thought of as operators in $\mathcal{H}_1,\mathcal{H}_2$, respectively, then we write $A_1\oplus A_2$ for their direct sum as an operator in $\mathcal{H}_1\oplus\mathcal{H}_2$).
    Suppose $$\mathcal{M}_i = (\Sigma, \Gamma, \mathcal{H}^{(i)}, U^{(i)}, M^{(i)})\ (i = 1, 2)$$ are two QMMs with the same input and output alphabets. Then the \textit{direct sum} of $\mathcal{M}_1$ and $\mathcal{M}_2$ is defined as
    \[
        \mathcal{M}_1 \oplus \mathcal{M}_2 = (\Sigma, \Gamma, \mathcal{H}_1 \oplus \mathcal{H}_2, U, M),
    \]
    where $$U = \{ U^{(1)}_\sigma \oplus U^{(2)}_\sigma: \sigma \in \Sigma \}\ {\rm and}\ M = \{ M^{(1)}_m \oplus M^{(2)}_m: m \in \Gamma \}.$$
    Obviously, $\mathcal{M}_1 \oplus \mathcal{M}_2$ is also a QMM.

    \begin{theorem} \label{thm-iqmm-eq}
        Given two QMMs $\mathcal{M}_1$ and $\mathcal{M}_2$ with state Hilbert spaces $\mathcal{H}_1$ and $\mathcal{H}_2$, respectively. Let $n_1 = \dim \mathcal{H}_1$ and $n_2 = \dim \mathcal{H}_2$.
        \begin{enumerate}
          \item The following statements are equivalent:
          \begin{enumerate}
            \item $(\mathcal{M}_1, \rho_1) \sim (\mathcal{M}_2, \rho_2)$.
            \item $\rho_1 \sim \rho_2$ in $\mathcal{M}_1 \oplus \mathcal{M}_2$.
            \item $(\mathcal{M}_1, \rho_1) \sim^{n_1^2 + n_2^2 - 1} (\mathcal{M}_2, \rho_2)$.
            \item $\rho_1 \sim^{n_1^2 + n_2^2 - 1} \rho_2$ in $\mathcal{M}_1 \oplus \mathcal{M}_2$.
          \end{enumerate}
          \item For every $k \in \mathbb{N}$, the following statements are equivalent:
          \begin{enumerate}
            \item $(\mathcal{M}_1, \rho_1) \sim_k (\mathcal{M}_2, \rho_2)$.
            \item $\rho_1 \sim_{k} \rho_2$ in $\mathcal{M}_1 \oplus \mathcal{M}_2$.
            \item $(\mathcal{M}_1, \rho_1) \sim_k^{n_1^2 + n_2^2 - 1} (\mathcal{M}_2, \rho_2)$.
            \item $\rho_1 \sim_{k}^{n_1^2 + n_2^2 - 1} \rho_2$ in $\mathcal{M}_1 \oplus \mathcal{M}_2$.
          \end{enumerate}
        \end{enumerate}
    \end{theorem}

If both $\mathcal{M}_1$ and $\mathcal{M}_2$ are real, then the experiment size $n_1^2 + n_2^2 - 1$ can be improved to $\frac 1 2 n_1(n_1+1) + \frac 1 2 n_2(n_2+1) - 1.$

    The above theorem implies that the algorithms in Theorem \ref{thm-algo} can be used for checking equivalence of two QMMs.

        \subsection{Minimization of machines}

        Finally, we consider the minimisation problem of QMMs. Formally, it can be formulated as the following decision problem:

          \begin{problem} \label{prob1} Given a QMM $\mathcal{M}^*$ and its initial state $\rho^*$, whether there is a QMM $\mathcal{M}$ and its initial state $\rho$ such that $\dim \mathcal{H} < \dim \mathcal{H}^*$ and $(\mathcal{M}, \rho) \sim (\mathcal{M}^*, \rho^*)$.\end{problem}

          A variant of this problem with a bound on the number of measurements is stated as:

           \begin{problem} \label{prob2} Given a QMM $\mathcal{M}^*$ and its initial state $\rho^*$ together with an integer $k$, whether there is a QMM $\mathcal{M}$ and its initial state $\rho$ such that $\dim \mathcal{H} < \dim \mathcal{H}^*$ and $(\mathcal{M}, \rho) \sim_k (\mathcal{M}^*, \rho^*)$.\end{problem}

Our result is then given as the following:
    \begin{theorem} \label{thm-decision}
    Both Problem \ref{prob1} and Problem \ref{prob2} are in \textbf{PSPACE}.
    \end{theorem}

    \subsection{Remarks}

    As mentioned in the Introduction, the equivalence checking problem of various quantum finite-state automata rather than QMMs has been thoroughly studied in the previous literature. The techniques in this paper are developed for quantum Mealy machines. However, they can also be used to improve the previous complexity results for quantum finite-states automata (see Table \ref{tab2}).

    \begin{itemize}\item For equivalence checking of two real-valued quantum automata (i.e. all entries of its unitary matrices and measurements are real numbers), our improvements on the length of inputs for equivalence checking are summarised in Table \ref{tab1}.
\iffalse
\item Our approach is even applicable to improve the complexity of equivalence checking for probabilistic automata. The improvements on complexities are summarised in Table \ref{tab2}.
\fi
\item It was proved in \cite{Mat12} that the minimization problem for several models of quantum finite-state automata (MO-QFA, MM-QFA, MO-gQFA) can be solved in \textbf{EXPTIME}. Our technique in proving Theorem \ref{thm-decision} can be used to improve this result to \textbf{PSPACE}.
\item It is worth pointing out that the results given in the previous literature (see the second columns of Tables \ref{tab1} and \ref{tab2}) were proved only in the case of pure states. However, our results (see the third columns of Tables \ref{tab1} and \ref{tab2}) are valid for the general case of mixed states.
%It is worth pointing out that only equivalence of pure states has been considered in the previous literature, but our results are valid for mixed states.

\end{itemize}
  \begin{table}[!hbp]
        \centering
            \begin{tabular} {|c|c|c|}
            \hline
            Model & $m$-equivalence & Our improvements \\
            \hline
            MO-QFA \cite{Moo00} \cite{Bro02} & $n_1^2+n_2^2-1$ \cite{Kos01} \cite{Li06} \cite{Li12} & $\frac 1 2 n_1(n_1+1)+\frac 1 2 n_2(n_2+1) -1$ \\
            \hline
            MM-QFA \cite{Kon97} \cite{Bro02} & $3n_1^2+3n_2^2-1$ \cite{Li08} & $\frac 3 2 n_1(n_1+1)+\frac 3 2 n_2(n_2+1) -1$ \\
            \hline
            CL-QFA \cite{Ber03} & $c_1n_1^2+c_2n_2^2-1$ \cite{Li08} & $\frac {c_1} 2 n_1(n_1+1)+\frac {c_2} 2 n_2(n_2+1)-1$ \\
            \hline
            QSM \cite{Gud00} \cite{Qiu02} & $n_1^2+n_2^2-1$ \cite{Li062} & $\frac 1 2 n_1(n_1+1)+\frac 1 2 n_2(n_2+1) -1$ \\
            \hline
            continuum QMM \cite{Wan18} & $n_1^2+n_2^2-1$ \cite{Wan18} & $\frac 1 2 n_1(n_1+1)+\frac 1 2 n_2(n_2+1)-1$ \\
            \hline
            \end{tabular}
            \caption{Shorter inputs for equivalence checking of two real-valued quantum automata, where $n_1$ and $n_2$ are the dimension of the state Hilbert spaces of the two automata, respectively. }
            \label{tab1}
        \end{table}
          \begin{table}[!hbp]
    \centering
        \begin{tabular} {|c|c|c|c|}
                \hline
        Model & Complexity & Our improvements \\
        \hline
        \iffalse
        PFA, PMM \cite{Paz71} & $O(n^4)$ \cite{Tze92} & $O(n^3)$ \\
        \hline
        \fi
        MO-QFA \cite{Moo00} \cite{Bro02} & $O(n^8)$ \cite{Kos01} \cite{Qiu11} & $O(n^6)$ \\
        \hline
        MM-QFA \cite{Kon97} \cite{Bro02} & $O(n^8)$ \cite{Li08} & $O(n^6)$ \\
        \hline
        CL-QFA \cite{Ber03} & $O((c_1n_1^2+c_2n_2^2)^4)$ \cite{Li08} & $O((c_1n_1^2+c_2n_2^2)^3)$ \\
        \hline
        QSM \cite{Gud00} \cite{Qiu02} & $O(n^{12})$ \cite{Li06} & $O(n^6)$ \\
        \hline
        continuum QMM \cite{Wan18} & $O(n^8)$ \cite{Wan18} & $O(n^6)$ \\
        \hline
        \end{tabular}
        \caption{Better complexities for checking equivalence of two automata, where $n_1$ and $n_2$ are the dimension of the state Hilbert spaces of the two automata, respectively, and $n = n_1+n_2$. }
        \label{tab2}
    \end{table}

    \section{Case Studies}\label{test}

To test the efficiency of Algorithms \ref{algo1} and \ref{algo2} presented in the last section, we prepared a set of benchmarks for case studies. It consists of 30 test cases (from test001 to test030), and the detailed descriptions of them can be found in \cite{urlbenchmark}. In this section, we only briefly discuss a couple of examples in order to give the reader a basic idea about them.

For better testing the efficiency of our algorithms, the state Hilbert spaces in these test cases are designed to be of various dimensions, e.g. test001 is of dimension 2 (a single qubit), while test017 is of dimension $2^5=32$ (5 qubits). The quantum Mealy machines and circuits in Example \ref{example1} are associated with test002 and test005, and Example \ref{example2} with test008, test009 and test010.

Algorithms \ref{algo1} and \ref{algo2} are implemented in C/C++ compiled by GCC 5.4.0. We test our algorithms on a Linux workstation: Intel(R) Xeon(R) CPU E7-8850 v2 2.30GHz with 24M Cache. All test cases utilize the single thread mode. To show the improvements displayed in Table \ref{tab2}, the experimental result is collected in Table \ref{tab3} with comparisons between the method of complexity $O(n^8)$, which can be directly derived from the techniques introduced in \cite{Tze92}, and our improved method of complexity $O(n^6)$. This table contains those test cases with large dimensions of the state Hilbert spaces.

\begin{table}
    \centering
    \begin{tabular} {|c|c|c|c|c|c|c|c|c|c|c|c|c|c|c|}
    \hline
    Tests & $n$ & $O(n^8)$ RT(s) & $O(n^6)$ RT(s) \\
    \hline
    test016 & 16 & 2.83 & 0.20 \\
    \hline
    test017 & 32 & 60.29 & 2.58 \\
    \hline
    test026 & 16 & 10.00 & 0.34 \\
    \hline
    test027 & 16 & 33.93 & 0.65 \\
    \hline
    test028 & 16 & 14.44 & 0.49 \\
    \hline
    test029 & 16 & 16.52 & 0.35 \\
    \hline
    \end{tabular}
    \caption{Experiment results. Here, $n$ stands for the dimension of the state Hilbert spaces of the QMM,
         $O(n^8)$ RT(s) for the running time of the method with complexity $O(n^8)$, and
         $O(n^6)$ RT(s) for the running time of the method with complexity $O(n^6)$.}
    \label{tab3}
\end{table}

    \section{Proofs of Theorems}\label{proofs}

    Now we arrive at the more technical part of this paper. In this section, we give the proofs of the theorems presented in Sec. \ref{main}. For readability, some tedious parts of these proofs are provided in the Appendices.

    \subsection{Proof of Theorem \ref{thm-state-eq}}

    First, we notice that Part 1 is a corollary of Part 2. If Part 2 holds, i.e.$$\rho_s \sim_{k} \rho_t \Longleftrightarrow \rho_s \sim_{k}^{n^2-1} \rho_t$$ for every $k \in \mathbb{N}$, then for all $k \geq n^2-1$, we have: $$\rho_s \sim_{k} \rho_t \Longleftrightarrow \rho_s \sim_{k}^{n^2-1} \rho_t \Longleftrightarrow \rho_s \sim_{n^2-1}^{n^2-1} \rho_t \Longleftrightarrow \rho_s \sim^{n^2-1} \rho_t,$$ which does not depend on $k$. This implies that if $\rho_s \sim^{n^2-1} \rho_t$, then $\rho_s \sim_{k} \rho_t$ for all $k \geq n^2-1$; that is, $\rho_s \sim \rho_t$. So, we only need to prove Part 2 (see \ref{app1} for a simple derivation of Part 1 from Part 2). Before doing it, we need some preparations.

    Let $\mathcal{M} = (\Sigma, \Gamma, \mathcal{H}, U, M)$, and let $\rho$ be an Hermitian operator. We define the set:
    \[
        \mathfrak{D}_k(\rho, m) = \{ \rho_{b|a, \mathcal{S}}^{\mathcal{M}}: a \in \Sigma^*, \mathcal{S} \in \mathfrak{S}_a, b \in \Gamma^{\abs{\mathcal{S}}}, \abs{a}+\abs{\mathcal{S}} \leq m, \abs{\mathcal{S}} \leq k \},
    \]
    where $$\rho_{b|a, \mathcal{S}}^{\mathcal{M}} = V_{b|a, \mathcal{S}} \rho V_{b|a, \mathcal{S}}^\dag\ {\rm and}\ V_{b|a, \mathcal{S}} = U_{a_{\abs{\mathcal{S}}+1}} M_{b_{\abs{\mathcal{S}}}} U_{a_{\abs{\mathcal{S}}}} \dots M_{b_1} U_{a_1}.$$
    Intuitively, this set records all of the possible states of the machine starting in state $\rho$ with the bounds $m$ on the experiment size and $k$ on the number of allowed measurements.
    Especially, $\mathfrak{D}_k(\rho, 0) = \{ \rho \}$ for all $k \in \mathbb{N}$. Obviously, the following properties hold: for every $m, k \in \mathbb{N}$,
    \begin{enumerate}
      \item $\mathfrak{D}_k(\rho, m) \subseteq \mathfrak{D}_k(\rho, m+1)$ and thus $\operatorname{span}\mathfrak{D}_k(\rho, m) \subseteq \operatorname{span}\mathfrak{D}_k(\rho, m+1)$.
      \item $\mathfrak{D}_k(\rho, m) \subseteq \mathfrak{D}_{k+1}(\rho, m)$ and thus $\operatorname{span}\mathfrak{D}_k(\rho, m) \subseteq \operatorname{span}\mathfrak{D}_{k+1}(\rho, m)$.
      \item $\dim \operatorname{span}\mathfrak{D}_k(\rho, m) \leq n^2$.
    \end{enumerate}

    Furthermore, we have the following:

    \begin{lemma} \label{prop-span-k}
        If $\operatorname{span} \mathfrak{D}_l(\rho, m) = \operatorname{span} \mathfrak{D}_l(\rho, m+1)$ for every $0 \leq l \leq k$, then for every $0 \leq l \leq k$ and $\delta \in \mathbb{N}$, $\operatorname{span} \mathfrak{D}_l(\rho, m) = \operatorname{span} \mathfrak{D}_l(\rho, m+\delta)$.
    \end{lemma}

    \begin{proof}
        We prove it by induction on $\delta$.

        \textbf{Basis}. It is trivial when $\delta = 1$.

        \textbf{Induction}. Suppose it is true for some $\delta \geq 1$ that $$\operatorname{span}\mathfrak{D}_l(\rho, m) = \operatorname{span}\mathfrak{D}_l(\rho, m+\delta)$$ for all $0 \leq l \leq k$. For every $\rho_{b|a, \mathcal{S}}^{\mathcal{M}} \in \mathfrak{D}_l(\rho, m+\delta+1)$ for some $a \in \Sigma^*$, and for all $\mathcal{S} \in \mathfrak{S}_a$ and $b \in \Gamma^{\abs{\mathcal{S}}}$ with $\abs{a}+\abs{\mathcal{S}} \leq m+\delta+1$ and $\abs{\mathcal{S}} \leq l$ for some $0 \leq l \leq k$, we consider the following two cases:

        \textbf{Case 1}. $\mathcal{S}$ is closed, i.e. $s_{\abs{\mathcal{S}}} = \abs{a}$: Then we set $\mathcal{S}^- = \{ s_1, s_2, \dots, s_{\abs{\mathcal{S}}-1} \}$ and $b^- = b[1:\abs{b}-1]$. It holds that $$\rho_{b|a, \mathcal{S}}^{\mathcal{M}} = M_{b[\abs{b}]} \rho_{b^-|a, \mathcal{S}^-}^{\mathcal{M}}  M_{b[\abs{b}]}^\dag.$$ By the assumption, we obtain: $$\rho_{b^-|a, \mathcal{S}^-}^\mathcal{M} \in \mathfrak{D}_{l-1}(\rho, m+\delta) \subseteq \operatorname{span} \mathfrak{D}_{l-1}(\rho, m+\delta) = \operatorname{span} \mathfrak{D}_{l-1}(\rho, m).$$ Thus, we have:
        \begin{align*}
            \rho_{b|a, \mathcal{S}}^{\mathcal{M}}
            & = M_{b[\abs{b}]} \rho_{b^-|a, \mathcal{S}^-}^{\mathcal{M}}  M_{b[\abs{b}]}^\dag
             \in M_{b[\abs{b}]} \left( \operatorname{span} \mathfrak{D}_{l-1}(\rho, m) \right) M_{b[\abs{b}]}^\dag \\
            & = \operatorname{span} \left( M_{b[\abs{b}]}\mathfrak{D}_{l-1}(\rho, m) M_{b[\abs{b}]}^\dag \right)\\
             &\subseteq \operatorname{span} \mathfrak{D}_{l}(\rho, m+1)
             = \operatorname{span} \mathfrak{D}_l(\rho, m).
        \end{align*}

        \textbf{Case 2}. $\mathcal{S}$ is not closed: Then put $a^- = a[1:\abs{a}-1]$. It follows that $$\rho_{b|a, \mathcal{S}}^{\mathcal{M}} = U_{a[\abs{a}]} \rho_{b|a^-, \mathcal{S}}^{\mathcal{M}} U_{a[\abs{a}]}^\dag.$$ By the assumption, it holds that $$\rho_{b|a^-, \mathcal{S}}^{\mathcal{M}} \in \mathfrak{D}_l(\rho, m+\delta) \subseteq \operatorname{span} \mathfrak{D}_l(\rho, m+\delta) = \operatorname{span} \mathfrak{D}_l(\rho, m).$$ So, we have:
        \begin{align*}
            \rho_{b|a, \mathcal{S}}^{\mathcal{M}}
            & = U_{a[\abs{a}]} \rho_{b|a^-, \mathcal{S}}^{\mathcal{M}} U_{a[\abs{a}]}^\dag
             \in U_{a[\abs{a}]} \left( \operatorname{span} \mathfrak{D}_l(\rho, m) \right) U_{a[\abs{a}]}^\dag \\
            & = \operatorname{span} \left( U_{a[\abs{a}]} \mathfrak{D}_l(\rho, m) U_{a[\abs{a}]}^\dag \right)\\
             &\subseteq \operatorname{span} \mathfrak{D}_l(\rho, m+1)
             = \operatorname{span} \mathfrak{D}_l(\rho, m).
        \end{align*}

      The above two cases together yield $\rho_{b|a, \mathcal{S}} \in \operatorname{span} \mathfrak{D}_l(\rho, m)$, and thus $$\operatorname{span} \mathfrak{D}_l(\rho, m+\delta+1) \subseteq \operatorname{span} \mathfrak{D}_l(\rho, m).$$ On the other hand, it is clear that $\operatorname{span} \mathfrak{D}_l(\rho, m) \subseteq \operatorname{span} \mathfrak{D}_l(\rho, m+\delta+1)$. We conclude that $\operatorname{span} \mathfrak{D}_l(\rho, m+\delta+1) = \operatorname{span} \mathfrak{D}_l(\rho, m)$ for every $0 \leq l \leq k$, and complete the proof of this lemma.
    \end{proof}

    \begin{lemma} \label{col-span-k}
        $\mathfrak{D}_k(\rho, n^2 - 1) \supseteq \mathfrak{D}_k(\rho, m)$ for every $k \in \mathbb{N}$ and $m \in \mathbb{N}$.
    \end{lemma}

    \begin{proof}
        Lemma \ref{prop-span-k} implies that for each $k \in \mathbb{N}$, there is an $m$ such that $\operatorname{span} \mathfrak{D}_l(\rho, m) = \operatorname{span} \mathfrak{D}_l(\rho, m+\delta)$ for every $\delta \in \mathbb{N}$ and $0 \leq l \leq k$. Let $$m_k = \min \{ m \in \mathbb{N}: \operatorname{span} \mathfrak{D}_l(\rho, m) = \operatorname{span} \mathfrak{D}_l(\rho, m+\delta), \forall \delta \in \mathbb{N}, 0 \leq l \leq k \}.$$Then $m_k \leq m_{k+1}$ for every $k \in \mathbb{N}$. Let us first prove that
        \begin{equation}\label{mid-in}
        \dim \operatorname{span} \mathfrak{D}_k(\rho, m_k) \geq m_k+1
        \end{equation}
        for every $k \in \mathbb{N}$ by induction.

        \textbf{Basis}. $k = 0$: By contradiction, we assume that $\dim \operatorname{span} \mathfrak{D}_0(\rho, m_0) < m_0+1$. Note that
        $
        \dim \operatorname{span} \mathfrak{D}_0(\rho, m) \leq \dim \operatorname{span} \mathfrak{D}_0(\rho, m+1)
        $
        for every $m \in \mathbb{N}$ and $\dim \operatorname{span} \mathfrak{D}_0(\rho, 0) = 1$. By the Pigeonhole Principle, there is a $0 \leq m' < m_0$ such that
        $
        \dim \operatorname{span} \mathfrak{D}_0(\rho, m') = \dim \operatorname{span} \mathfrak{D}_0(\rho, m'+1),
        $
        which conflicts with the minimality of $m_0$. Hence, $\dim \operatorname{span} \mathfrak{D}_0(\rho, m_0) \geq m_0+1$.

        \textbf{Induction}. Suppose inequality (\ref{mid-in}) is true for $k \geq 0$. By contradiction, we assume that $\dim \operatorname{span} \mathfrak{D}_{k+1}(\rho, m_{k+1}) < m_{k+1}+1$. Since $m_k \leq m_{k+1}$, we have:
        \begin{align*}
            m_{k+1}+1 & > \dim \operatorname{span} \mathfrak{D}_{k+1}(\rho, m_{k+1}) \geq \dim \operatorname{span} \mathfrak{D}_{k+1}(\rho, m_k) \\
            & \geq \dim \operatorname{span} \mathfrak{D}_k(\rho, m_k) \geq m_k+1.
        \end{align*}
        By the Pigeonhole Principle, there is a $m_k \leq m' < m_{k+1}$ such that
        $$
        \dim \operatorname{span} \mathfrak{D}_{k+1}(\rho, m') = \dim \operatorname{span} \mathfrak{D}_{k+1}(\rho, m'+1),
        $$
        which conflicts with the minimality of $m_{k+1}$. Hence, $\dim \operatorname{span} \mathfrak{D}_{k+1}(\rho, m_{k+1}) \geq m_{k+1}+1$, and we complete the proof of (\ref{mid-in}).

Using (\ref{mid-in}), we see that $m_k+1 \leq \dim \operatorname{span} \mathfrak{D}_{k}(\rho, m_{k}) \leq n^2$ for every $k \in \mathbb{N}$, and $m_k \leq n^2-1$. Then we proved the lemma.
    \end{proof}

    Now we are ready to prove Part 2: $\rho_s \sim_{k} \rho_t \Longleftrightarrow \rho_s \sim_{k}^{n^2-1} \rho_t$. Clearly, we only need to prove the ``if'' part.
        Suppose that $\rho_s \sim_{k}^{n^2-1} \rho_t$. Then
        $$
            \Pr\nolimits_{\rho_s}^{\mathcal{M}}(b|a, \mathcal{S}) = \Pr\nolimits_{\rho_t}^{\mathcal{M}}(b|a, \mathcal{S})
        $$
        for every $a \in \Sigma^*$, $\mathcal{S} \in \mathfrak{S}_a$ and $b \in \Gamma^{\abs{\mathcal{S}}}$ with $\abs{a}+\abs{\mathcal{S}} \leq n^2 - 1$ and $\abs{\mathcal{S}} \leq k$. We conclude that
        $
            \operatorname{tr}(\rho_{b|a, \mathcal{S}}^{\mathcal{M}}) = 0
        $
        for every $\rho_{b|a, \mathcal{S}}^{\mathcal{M}} \in \mathfrak{D}_k(\rho, n^2-1)$, where $\rho = \rho_s - \rho_t$.

        On the other hand, for every $a \in \Sigma^*$, $\mathcal{S} \in \mathfrak{S}_a$ and $b \in \Gamma^{\abs{\mathcal{S}}}$ with $\abs{\mathcal{S}} \leq k$, by Lemma \ref{col-span-k}, we have $\rho_{b|a, \mathcal{S}}^{\mathcal{M}} \in \operatorname{span} \mathfrak{D}_k(\rho, n^2-1)$, and then $\operatorname{tr}(\rho_{b|a, \mathcal{S}}^{\mathcal{M}}) = 0$,
        i.e. $$\Pr\nolimits_{\rho_s}^{\mathcal{M}}(b|a, \mathcal{S}) = \Pr\nolimits_{\rho_t}^{\mathcal{M}}(b|a, \mathcal{S}),$$ which immediately yields $\rho_s \sim_{k} \rho_t$. Therefore, we complete the proof for the general case.

        For the case that $\mathcal{M}$ is real, however, the key observation is that $\operatorname{tr}(\rho) = \operatorname{tr}(\text{Re}(\rho))$ if $\rho$ is Hermitian. Thus, we only need to consider the real part of the density operators. Define:
    \[
        \mathfrak{D}_k^{\text{Re}}(\rho, m) = \{ \text{Re} (\rho_{b|a, \mathcal{S}}^{\mathcal{M}}): a \in \Sigma^*, \mathcal{S} \in \mathfrak{S}_a, b \in \Gamma^{\abs{\mathcal{S}}}, \abs{a}+\abs{\mathcal{S}} \leq m, \abs{\mathcal{S}} \leq k \},
    \]
    where $\text{Re}(x)$ denotes the real part of $x$, e.g. $\text{Re}(3+4i) = 3$.
    We have a better bound:
    \[
    \dim \operatorname{span}\mathfrak{D}_k^{\text{Re}}(\rho, m) \leq \frac 1 2 n(n+1)
    \]
    for every $m, k \in \mathbb{N}$ if $\rho$ is Hermitian. Note that $\rho$ need not be real. Following the idea of the above proof, we obtain $\rho_s \sim_{k} \rho_t \Longleftrightarrow \rho_s \sim_{k}^{\frac 1 2 n(n+1)-1} \rho_t$ if $\mathcal{M}$ is real.

\iffalse
    The reason why the previous proof still works for the case that $\mathcal{M}$ is real is mainly due to the fact that the imaginary part won't influence the real part, and thus Proposition \ref{prop-span-k} still holds. The key step is that, for example,
    \begin{align*}
            \text{Re} (\rho_{b|a, \mathcal{S}}^{\mathcal{M}})
            & = \text{Re} (M_{b[\abs{b}]} \rho_{b^-|a, \mathcal{S}^-}^{\mathcal{M}}  M_{b[\abs{b}]}^\dag) \\
            & = M_{b[\abs{b}]} \text{Re}(\rho_{b^-|a, \mathcal{S}^-}^{\mathcal{M}})  M_{b[\abs{b}]}^\dag \\
            & = M_{b[\abs{b}]} \left( \sum_{a', \mathcal{S}', b': \rho_{b'|a', \mathcal{S}'}^{\mathcal{M}} \in \mathfrak{D}_{l-1}^{\text{Re}}(\rho, m)} \alpha_{b'|a', \mathcal{S}'} \text{Re}(\rho_{b'|a', \mathcal{S}'}^{\mathcal{M}}) \right) M_{b[\abs{b}]}^\dag \\
            & = \sum_{a', \mathcal{S}', b': \rho_{b'|a', \mathcal{S}'}^{\mathcal{M}} \in \mathfrak{D}_{l-1}(\rho, m)} \alpha_{b'|a', \mathcal{S}'} \text{Re} (M_{b[\abs{b}]} \rho_{b'|a', \mathcal{S}'}^{\mathcal{M}}  M_{b[\abs{b}]}^\dag) \\
            & = \sum_{a', \mathcal{S}', b': \rho_{b'|a', \mathcal{S}'}^{\mathcal{M}} \in \mathfrak{D}_{l-1}(\rho, m)} \alpha_{b'|(a', \mathcal{S}')} \text{Re}(\rho_{b'b[\abs{b}]|(a', \mathcal{S}' \cup \{ \abs{a} \})}^{\mathcal{M}}) \\
            & \in \operatorname{span} \mathfrak{D}_l^{\text{Re}}(\rho, m+1) = \operatorname{span} \mathfrak{D}_l^{\text{Re}}(\rho, m).
        \end{align*}
\fi

    \subsection{Proof of Theorem \ref{thm-algo}}

    The proof of correctness of Algorithms \ref{algo1} and \ref{algo2} is put in \ref{app2}.
Here, we analyze their complexities. We only consider Algorithm \ref{algo2}, and Algorithm \ref{algo1} can be analyzed similarly.

    Since $\dim \operatorname{span} \mathfrak{D}_l(\rho, n^2-1) \leq n^2$, we have $\abs{\mathfrak{B}_l} \leq n^2$ for $0 \leq l \leq k$. For each element $\rho_{b|(a, \mathcal{S})} \in \mathfrak{B}_l$, when $\rho_{b|(a, \mathcal{S})}$ is added into $\mathfrak{B}_l$ in the algorithm, there are $m = \abs{\Sigma}+\abs{\Gamma}$ tuples that are pushed into $Q$. Thus, there are at most $$\sum_{l=0}^k m \abs{\mathfrak{B}_l} = O(kmn^2)$$ tuples that are pushed into $Q$ in total. In each iteration of the ``while'' loop, we have to check whether an operator is linearly independent to a set of operators (whether $\rho_{b|(a, \mathcal{S})} \in \operatorname{span} \mathfrak{B}_l$, see Line 6 and 7 in Algorithm \ref{algo2}). Note that there are some simple methods, e.g. Gaussian Elimination, to check whether $n$ vectors in a $d$-dimensional space are linearly independent in $O(dn^2)$ time. The ``while'' loop of Algorithm \ref{algo2} can be summarized as follows:
    \begin{enumerate}
      \item Pop the front element $(a, \mathcal{S}, b)$ from $Q$ and calculate $\rho_{b|(a, \mathcal{S})}$.
      \item Check whether $\mathfrak{B}_l \cup \{ \rho_{b|(a, \mathcal{S})} \}$ is linearly independent for some $0 \leq l \leq k$. Each check costs $O(n^6)$ time.
      \item Add $\rho_{b|(a, \mathcal{S})}$ into $\mathfrak{B}_l$ for some $0 \leq l \leq k$ in $O(1)$ time (at most $O(kn^2)$ times).
      \item Push new tuples $(a\sigma, \mathcal{S}, b)$ for $\sigma \in \Sigma$ and $(a, \mathcal{S}+\{\abs{a}\}, b\gamma)$ for $\gamma \in \Gamma$ into $Q$ (at most $O(kn^2)$ times).
    \end{enumerate}
    It is clear that the overall complexity is $O(kmn^8)$. In fact, the complexity can be reduced to $O(kmn^6)$ if we adapt the trick used in \cite{Kie11}. We see that the bottleneck is to check whether an operator is linearly independent to a set of operators which changes not so often. Another observation is that whence an operator $\varrho$ is checked to be linearly independent to $\mathfrak{B}$, the only operation is to add it into $\mathfrak{B}$. We could make the two things mentioned above more ``balanced'' in time.
    To improve the time complexity, we introduce the inner product of operators $A$ and $B$ defined by:
    \[
        \langle A, B \rangle = \operatorname{tr}(A^\dag B) = \sum_{i=1}^n \sum_{j=1}^n A^*_{ij} B_{ij},
    \]
    where $c^*$ is the conjugate of a complex number $c$. It needs $O(n^2)$ time to compute.
    We use a so-called ``lazy'' Gram-Schmidt process to maintain the orthogonal set $\mathcal{O}$ with respect to $\mathfrak{B}$ such that $\operatorname{span} \mathfrak B = \operatorname{span} \mathcal O$, as follows:
    \begin{enumerate}
      \item Initially, $\mathfrak{B} = \emptyset$. Set $\mathcal{O} = \emptyset$.
      \item When checking whether an operator $\varrho$ is linearly independent to $\mathfrak B$, we only need to check whether $\varrho$ is linearly independent to $\mathcal O$. Because $\mathcal O$ is an orthogonal set, i.e. all elements in $\mathcal O$ are pairwise orthogonal, we conclude that $\varrho$ is not linearly independent to $\mathcal O$ if and only if
          \begin{equation} \label{eq3}
            \langle \varrho, \varrho \rangle = \sum_{\varrho' \in \mathcal O} \frac {\abs{\langle \varrho', \varrho \rangle}^2} {\langle \varrho', \varrho' \rangle},
          \end{equation}
          which needs $O(n^4)$ time to check.
      \item When $\varrho$ is linearly independent to $\mathfrak B$, as well as $\mathcal O$, we have to add it into $\mathfrak B$ and maintain $\mathcal O$ to meet $\operatorname{span} \mathfrak B = \operatorname{span} \mathcal O$. Let
          \begin{equation} \label{eq4}
            \hat\varrho = \varrho - \sum_{\varrho' \in \mathcal O} \frac {\langle \varrho', \varrho \rangle} {\langle \varrho', \varrho' \rangle} \varrho',
          \end{equation}
          then $\hat\varrho$ is orthogonal to $\mathcal O$ and add $\hat\varrho$ into $\mathcal{O}$, which needs $O(n^4)$ time to compute.
    \end{enumerate}
    With this observation, the ``while'' loop of Algorithm \ref{algo2} can be modified and summarized as follows:
    \begin{enumerate}
      \item Check whether $\mathfrak{B}_l \cup \{ \rho_{b|(a, \mathcal{S})} \}$ is linearly independent using Eq. (\ref{eq3}) in $O(n^4)$ time.
      \item Add $\rho_{b|(a, \mathcal{S})}$ into $\mathfrak{B}_l$ using Eq. (\ref{eq4}) in $O(n^4)$ time (at most $n^2-1 = O(n^2)$ times).
    \end{enumerate}
    It is clear that the overall complexity is now reduced to $O(kmn^6)$.

    \subsection{Proof of Theorem \ref{thm-iqmm-eq}}

    This theorem is established based on Theorem \ref{thm-state-eq}. Let
    \[
        \mathfrak{D}^\mathcal{M}(\rho, m) = \{ \rho_{b|(a, \mathcal{S})}^{\mathcal{M}}: a \in \Sigma^*, \mathcal{S} \in \mathfrak{S}_a, b \in \Gamma^{\abs{\mathcal{S}}}, \abs{a}+\abs{\mathcal{S}} \leq m \},
    \]
    The key observation is that
    $ \mathfrak{D}^{\mathcal{M}_1 \oplus \mathcal{M}_2}(\rho_1 \oplus 0 - 0 \oplus \rho_2, m) $
    is diagonal, and thus
    \[
        \dim \operatorname{span} \mathfrak{D}^{\mathcal{M}_1 \oplus \mathcal{M}_2}(\rho_1 \oplus 0 - 0 \oplus \rho_2, m) \leq n_1^2+n_2^2.
    \]
    Then we can prove the theorem with the same techniques used in the proof of Theorem \ref{thm-state-eq}.

    \subsection{Proof of Theorem \ref{thm-decision}}

  We first note that Theorem \ref{thm-state-eq} implies that $\rho_s \sim \rho_t \Longleftrightarrow \rho_s \sim_{k} \rho_t$ for every $k \geq n^2-1$. Thus, Problem \ref{prob1} is a special case of Problem \ref{prob2} for $k = n^2-1$. Now we may assume that $k \leq n^2-1$ and only consider Problem \ref{prob2}.

    Let $\vec\rho$ be the vectorization of $\rho$, which is an $n^2$-dimensional vector with entries $\vec\rho_{(i-1)n+j} = \rho_{ij}$. For an $n \times n$ matrix $M$, we define an $n^2 \times n^2$ matrix $\hat M$ with entries:
    \[
        \hat M_{(i-1)n+j, (x-1)n+y} = M_{ix} M_{jy}^*.
    \]
    Then it holds that $\overrightarrow{M \rho M^\dag} = \hat M \vec \rho$. Moreover, let $\eta$ be the vectorization of trace, which is an $n^2$-dimensional vector $\eta_{(i-1)n+j} = \delta_{ij}$, where $$\delta_{ij} = \begin{cases} 1 & i = j \\ 0 & i \neq j \end{cases}$$ is the Kronecker delta. It is clear that $\operatorname{tr}(M \rho M^\dag) = \eta^\dag \hat M \vec \rho$.

    We further have the following proposition as a quantum analog of Proposition 10 in \cite{Kie14}:
    \begin{proposition} \label{prop-eq-k}
        Let $\mathcal{M}_1 = (\Sigma, \Gamma, \mathcal{H}^{(1)}, U^{(1)}, M^{(1)})$ and $\mathcal{M}_2 = (\Sigma, \Gamma, \mathcal{H}^{(2)}, U^{(2)}, M^{(2)})$ be two QMMs with initial states $\rho_1$ and $\rho_2$, respectively, let $k$ be a positive integer, and let $n_1 = \dim \mathcal{H}_1$, $n_2 = \dim \mathcal{H}_2$ and $n = n_1^2+n_2^2$. Then
        $(\mathcal{M}_1, \rho_1) \sim_k (\mathcal{M}_2, \rho_2)$ if and only if there are $n \times n$ matrices $F^{(0)}, F^{(1)}, \dots, F^{(k)}$, $A_\sigma^{(0)}, A_\sigma^{(1)}, \dots, A_\sigma^{(k)}$ for every $\sigma \in \Sigma$ and $A_\gamma^{(0)}, A_\gamma^{(1)}, \dots, A_\gamma^{(k-1)}$ for every $\gamma \in \Gamma$ such that
        \begin{enumerate}
          \item $F_{\cdot, 1}^{(0)} = \begin{bmatrix} \vec \rho_1 \\ \vec \rho_2 \end{bmatrix}$.
          \item $\eta^\dag F^{(l)} = 0$ for $0 \leq l \leq k$, where $\eta = \begin{bmatrix} \eta_1 \\ -\eta_2 \end{bmatrix}$, $\eta_1$ and $\eta_2$ are the vectorizations of trace for $\mathcal{M}_1$ and $\mathcal{M}_2$, respectively.
          \item For every $\sigma \in \Sigma$,
          \[
            \begin{bmatrix}
                \hat U_\sigma^{(1)} & 0 \\
                0 & \hat U_\sigma^{(2)}
            \end{bmatrix} F^{(l)} = F^{(l)} A_\sigma^{(l)}
          \]
          for $0 \leq l \leq k$.
          \item For every $\gamma \in \Gamma$,
          \[
            \begin{bmatrix}
                \hat M_\gamma^{(1)} & 0 \\
                0 & \hat M_\gamma^{(2)}
            \end{bmatrix} F^{(l)} = F^{(l+1)} A_\gamma^{(l)}
          \]
          for $0 \leq l < k$.
        \end{enumerate}
    \end{proposition}

    \begin{proof}
``$\Longrightarrow$'' If $(\mathcal{M}_1, \rho_1) \sim_k (\mathcal{M}_2, \rho_2)$, then
\[
    \operatorname{tr}((\rho_1)^{\mathcal{M}_1}_{b|a,\mathcal{S}}) = \operatorname{tr}((\rho_2)^{\mathcal{M}_2}_{b|a,\mathcal{S}})
\]
for every $a \in \Sigma^*$, $\mathcal{S} \in \mathfrak{S}_a$ and $b \in \Gamma^{\abs{\mathcal{S}}}$ with $\abs{\mathcal{S}} \leq k$. Let $\rho = \rho_1 \oplus \rho_2$ and
\[
    \mathfrak{D}_l(\rho, m) = \{ \rho_{b|a,\mathcal{S}}: a \in \Sigma^*, \mathcal{S} \in \mathfrak{S}_a, b \in \Gamma^{\abs{\mathcal{S}}}, \abs{\mathcal{S}} \leq l \}.
\]
Then
\[
    \rho^{\mathcal{M}_1 \oplus \mathcal{M}_2}_{b|a,\mathcal{S}} = (\rho_1)^{\mathcal{M}_1}_{b|a,\mathcal{S}} \oplus (\rho_2)^{\mathcal{M}_2}_{b|a,\mathcal{S}}.
\]
By Theorem \ref{thm-iqmm-eq}, we have: $$(\mathcal{M}_1, \rho_1) \sim_k (\mathcal{M}_2, \rho_2) \Longleftrightarrow (\mathcal{M}_1, \rho_1) \sim_k^{n-1} (\mathcal{M}_2, \rho_2).$$ Proof of Theorem \ref{thm-state-eq} Part 2 reveals that $$\mathfrak{D}_0(\rho, n-1) \subseteq \mathfrak{D}_1(\rho, n-1) \subseteq \dots \subseteq \mathfrak{D}_k(\rho, n-1).$$ And note that $\dim \operatorname{span} \mathfrak{D}_k(\rho, n-1) \leq n$.
For every $0 \leq l \leq k$, let $\rho_l^{(1)}, \dots, \rho_l^{(n)} \in \mathfrak{D}_l(\rho, n-1)$ with $\operatorname{span} \mathfrak{D}_l(\rho, n-1) = \operatorname{span} \{ \rho_l^{(1)}, \dots, \rho_l^{(n)} \}$. In particular, guarantee that $\rho_0^{(1)} = \rho$.
Let $\rho_l^{(i)} = (\rho_1)_l^{(i)} \oplus (\rho_2)_l^{(i)}$ for $1 \leq i \leq n$. Then $\operatorname{tr}((\rho_1)_l^{(i)}) = \operatorname{tr}((\rho_2)_l^{(i)})$. We set $$f_i^{(l)} = \begin{bmatrix} (\vec \rho_1)_l^{(i)} \\ (\vec \rho_2)_l^{(i)} \end{bmatrix}$$ and $F^{(l)} = \begin{bmatrix} f_1^{(l)} \dots f_{n}^{(l)} \end{bmatrix}$.

\textbf{Part 1}. It is easy to verify that $$F^{(0)}_{\cdot, 1} = f_1^{(0)} = \begin{bmatrix} \vec \rho_1 \\ \vec \rho_2 \end{bmatrix}.$$

\textbf{Part 2}. We have: $$\eta^\dag f_i^{(l)} = \eta_1^\dag (\vec \rho_1)_l^{(i)} - \eta_2^\dag (\vec \rho_2)_l^{(i)} = \operatorname{tr}((\rho_1)_l^{(i)}) - \operatorname{tr}((\rho_2)_l^{(i)}) = 0.$$ Thus $\eta^\dag F^{(l)} = 0$.

{\vskip 4pt}

\textbf{Part 3}. If $\rho_{b|a, \mathcal{S}} \in \mathfrak{D}_l(\rho, n-1)$, then $\rho_{b|a\sigma, \mathcal{S}} \in \operatorname{span} \mathfrak{D}_l(\rho, n-1)$ for every $\sigma \in \Sigma$. Thus for every $0 \leq l \leq k$, $1 \leq i \leq n$ and $\sigma \in \Sigma$, there are coefficients $\alpha_{ij}$ such that
\begin{align*}
    \begin{bmatrix}
        \hat U^{(1)}_\sigma & 0 \\
        0 & \hat U^{(2)}_\sigma
    \end{bmatrix} f_i^{(l)}
    & = \begin{bmatrix}
        \hat U^{(1)}_\sigma (\vec \rho_1)_l^{(i)} \\
        \hat U^{(2)}_\sigma (\vec \rho_2)_l^{(i)}
    \end{bmatrix} \\
    & = \sum_j \alpha_{ij} \begin{bmatrix}
        (\vec \rho_1)_l^{(i)} \\
        (\vec \rho_2)_l^{(i)}
    \end{bmatrix} \\
    & = \sum_j \alpha_{ij} f_j^{(l)} \\
    & = \begin{bmatrix} f_1^{(l)} \dots f_{n}^{(l)} \end{bmatrix} \begin{bmatrix}
        \alpha_{i1} \\
        \dots \\
        \alpha_{in}
    \end{bmatrix} \\
    & = F^{(l)} \alpha_i.
\end{align*}
Set $A_\sigma^{(l)} = \begin{bmatrix} \alpha_1 & \dots & \alpha_n \end{bmatrix}$.

{\vskip 4pt}

\textbf{Part 4}. It holds similarly to Part 3.

``$\Longleftarrow$''. It can be proved by induction that for every $a \in \Sigma^*$, $\mathcal{S} \in \mathfrak{S}_a$ and $b \in \Gamma^{\abs{\mathcal{S}}}$ with $\abs{\mathcal{S}} \leq k$,
\[
    \begin{bmatrix}
        \hat V_{b|a,\mathcal{S}}^{(1)} & 0 \\
        0 & \hat V_{b|a,\mathcal{S}}^{(2)}
    \end{bmatrix} F^{(0)} = F^{(\abs{\mathcal{S}})} A_{b|a,\mathcal{S}},
\]
where $$A_{b|a,\mathcal{S}} = A_{a_1}^{(0)} A_{b_1}^{(0)} A_{a_2}^{(1)} A_{b_2}^{(1)} \dots A_{a_{\abs{\mathcal{S}}}}^{(\abs{\mathcal{S}}-1)} A_{b_{\abs{\mathcal{S}}}}^{(\abs{\mathcal{S}}-1)} A_{a_{\abs{\mathcal{S}}+1}}^{(\abs{\mathcal{S}})}.$$
Then
\begin{align*}
    \operatorname{tr}((\rho_1)_{b|a,\mathcal{S}}) - \operatorname{tr}((\rho_2)_{b|a,\mathcal{S}})
    & = \eta_1^\dag (\vec \rho_1)_{b|a,\mathcal{S}} - \eta_2^\dag (\vec \rho_2)_{b|a,\mathcal{S}} \\
    & = \begin{bmatrix}
        \eta_1 \\
        -\eta_2
    \end{bmatrix}^\dag \begin{bmatrix}
        \hat V_{b|a,\mathcal{S}}^{(1)} & 0 \\
        0 & \hat V_{b|a,\mathcal{S}}^{(2)}
    \end{bmatrix} \begin{bmatrix}
        \vec \rho_1 \\
        \vec \rho_2
    \end{bmatrix} \\
    & = \eta^\dag \begin{bmatrix}
        \hat V_{b|a,\mathcal{S}}^{(1)} & 0 \\
        0 & \hat V_{b|a,\mathcal{S}}^{(2)}
    \end{bmatrix} F^{(0)} e_1 \\
    & = \eta^\dag F^{(\abs{\mathcal{S}})} A_{b|a,\mathcal{S}} e_1 \\
    & = 0,
\end{align*}
where $e_1 = (1, 0, 0, \dots, 0)^T$.
\end{proof}

    We also need the following theorem from \cite{Can88}, \cite{Ren88}, \cite{Ren92} and \cite{Bas06}:
    \begin{theorem}\label{thm-existence}
        Given a set $\mathcal{P} = \{f_1, \dots, f_m\}$ of $m$ polynomials of degree $d$ in $n$ variables $x = (x_1, \dots, x_n)$. Let $\phi(\mathcal{P}, x)$ is a Boolean function of inequalities of the form $f_i(x) > 0$ or $f_i(x) \geq 0$, and let $S = \{ x \in \mathbb{R}^{n}: \phi(\mathcal{P}, x) \}$. Then:\begin{enumerate}\item there is an algorithm to decide whether $S = \emptyset$ in \textbf{PSPACE} \cite{Can88}, \cite{Ren92}. Moreover, it can be decided in $(md)^{O(n)}$ time \cite{Ren88}, and \item if $S \neq \emptyset$ then a sample $x \in S$ can be found in $\tau d^{O(n)}$ space \cite{Bas06}, where $\tau$ is the size of the coefficients of the polynomials.\end{enumerate}
    \end{theorem}

   Now we are ready to prove the Theorem \ref{thm-decision}.
   %Here, we only give \textit{an outline of the proof}, but leave the details to the Appendix.
   The conditions of Proposition \ref{prop-eq-k} on $\mathcal{M}_2$, including that it be a QMM, can be phrased in $O((\abs{\Sigma}+\abs{\Gamma})k(n_1+n_2)^4)$ polynomials of degree $d = 3$ in $O((\abs{\Sigma}+\abs{\Gamma})k(n_1+n_2)^4)$ variables. By Theorem \ref{thm-existence}, Problem \ref{prob2} is solvable in \textbf{PSPACE}.
   %Moreover, the exact QMM $\mathcal{M}$ and state $\rho$ can be computed in \textbf{EXPSPACE}.

    It is noted that Problem \ref{prob1} can be phrased in $O((\abs{\Sigma}+\abs{\Gamma})(n_1+n_2)^6)$ polynomials of degree $d = 3$ in $O((\abs{\Sigma}+\abs{\Gamma})(n_1+n_2)^6)$ variables with $k = O((n_1+n_2)^2)$. In fact, we can make it more efficient in $O((\abs{\Sigma}+\abs{\Gamma})(n_1+n_2)^4)$ polynomials of degree $d = 3$ in $O((\abs{\Sigma}+\abs{\Gamma})(n_1+n_2)^4)$ variables (see \ref{app3} for more details).

\section{Conclusion}\label{conclusion}

To offer effective tools for verification of quantum circuits, we define the model of quantum Mealy machines. Two efficient algorithms for checking equivalence of two states in the same quantum Mealy machines and for checking equivalence of two quantum Mealy machines are developed. We also prove that the minimisation problem for quantum Mealy machines can be solved in \textbf{PSPACE}.

Further future research, we plan to extend the ideas introduced and the results obtained in this paper along the following two lines:\begin{itemize}\item Study the equivalence checking problem for quantum programs, which are much harder to deal with, in particular in the case where loops and recursion are present \cite{Ying16}.
\item Incorporate the techniques developed in this paper with those in the previous work on model-checking of quantum systems \cite{YL14}, \cite{YSG} so that they can be applied to larger quantum circuits or more complicated properties than equivalence.
\end{itemize}
\iffalse
\section*{Acknowledgement}

The work presented in this paper has been partially supported by the National Key R\&D Program of China (Grant No. 2018YFA0306701), the National Natural Science Foundation of China (Grant No. 61832015), the Key Research Program of Frontier Sciences, Chinese Academy of Sciences (Grant No. QYZDJ-SSW-SYS003) and the Australian Research Council (Grant No. DP180100691).
\fi

\appendix

\renewcommand{\thetheorem}{\Alph{section}.\arabic{theorem}}

\section{A simple proof of Part 1 of Theorem \ref{thm-state-eq}} \label{app1}

  Part 1 of Theorem \ref{thm-state-eq} is a corollary of Part 2 of the same theorem. Here, we provide a simple and direct proof of it. Let $\mathcal{M} = (\Sigma, \Gamma, \mathcal{H}, U, M)$, and let $\rho$ be an Hermitian operator. Define
    \[
        \mathfrak{D}(\rho, m) = \{ \rho_{b|a, \mathcal{S}}^{\mathcal{M}}: a \in \Sigma^*, \mathcal{S} \in \mathfrak{S}_a, b \in \Gamma^{\abs{\mathcal{S}}}\ {\rm with}\ \abs{a}+\abs{\mathcal{S}} \leq m \},
    \]
    where $$\rho_{b|a, \mathcal{S}}^{\mathcal{M}} = V_{b|a, \mathcal{S}} \rho V_{b|a, \mathcal{S}}^\dag\ {\rm and}\ V_{b|a, \mathcal{S}} = U_{a_{\abs{\mathcal{S}}+1}} M_{b_{\abs{\mathcal{S}}}} U_{a_{\abs{\mathcal{S}}}} \dots M_{b_1} U_{a_1}.$$
    Especially, $\mathfrak{D}(\rho, 0) = \{ \rho \}$. Then it is easy to see that for every $m \in \mathbb{N}$,
    \begin{enumerate}
      \item $\mathfrak{D}(\rho, m) \subseteq \mathfrak{D}(\rho, m+1)$ and thus $\operatorname{span}\mathfrak{D}(\rho, m) \subseteq \operatorname{span}\mathfrak{D}(\rho, m+1)$.
      \item $\dim \operatorname{span}\mathfrak{D}(\rho, m) \leq n^2$.
    \end{enumerate} Furthermore, we have:
    \begin{proposition} \label{prop-span}
        If $\operatorname{span}\mathfrak{D}(\rho, m) = \operatorname{span}\mathfrak{D}(\rho, m+1)$ for some $m \in \mathbb{N}$, then $\operatorname{span}\mathfrak{D}(\rho, m) = \operatorname{span}\mathfrak{D}(\rho, m+k)$ for every $k \in \mathbb{N}$.
    \end{proposition}
    \begin{proof}
        We prove it by induction on $k$.

        {\vskip 4pt}

        \textbf{Basis}. It is trivial when $k = 1$.

        {\vskip 4pt}

        \textbf{Induction}. Suppose it is true for some $k \geq 1$ that $$\operatorname{span}\mathfrak{D}(\rho, m) = \operatorname{span}\mathfrak{D}(\rho, m+k).$$ For every $\rho_{b|a, \mathcal{S}}^{\mathcal{M}} \in \mathfrak{D}(\rho, m+k+1)$ for some $a \in \Sigma^*$, $\mathcal{S} \in \mathfrak{S}_a$ and $b \in \Gamma^{\abs{\mathcal{S}}}$ with $\abs{a}+\abs{\mathcal{S}} \leq m+k+1$, we consider the following two cases:

{\vskip 4pt}

        \textbf{Case 1}. $\mathcal{S}$ is closed, i.e. $s_{\abs{\mathcal{S}}} = \abs{a}$: Let $\mathcal{S}^- = \{ s_1, s_2, \dots, s_{\abs{\mathcal{S}}-1} \}$ and $b^- = b[1:\abs{b}-1]$. Then $$\rho_{b|a, \mathcal{S}}^{\mathcal{M}} = M_{b[\abs{b}]} \rho_{b^-|a, \mathcal{S}^-}^{\mathcal{M}}  M_{b[\abs{b}]}^\dag.$$ By the assumption, $\rho_{b^-|a, \mathcal{S}^-}^\mathcal{M} \in \mathfrak{D}(\rho, m+k) \subseteq \operatorname{span} \mathfrak{D}(\rho, m+k) = \operatorname{span} \mathfrak{D}(\rho, m)$. We have
        \begin{align*}
            \rho_{b|a, \mathcal{S}}^{\mathcal{M}}
            & = M_{b[\abs{b}]} \rho_{b^-|a, \mathcal{S}^-}^{\mathcal{M}}  M_{b[\abs{b}]}^\dag \\
            & \in M_{b[\abs{b}]} \left( \operatorname{span} \mathfrak{D}(\rho, m) \right) M_{b[\abs{b}]}^\dag \\
            & = \operatorname{span} \left( M_{b[\abs{b}]} \mathfrak{D}(\rho, m) M_{b[\abs{b}]}^\dag \right) \\
            & \subseteq \operatorname{span} \mathfrak{D}(\rho, m+1) = \operatorname{span} \mathfrak{D}(\rho, m).
        \end{align*}

        \textbf{Case 2}. $\mathcal{S}$ is not closed: Let $a^- = a[1:\abs{a}-1]$. Then $$\rho_{b|a, \mathcal{S}}^{\mathcal{M}} = U_{a[\abs{a}]} \rho_{b|a^-, \mathcal{S}}^{\mathcal{M}} U_{a[\abs{a}]}^\dag.$$ By the assumption, $\rho_{b|a^-, \mathcal{S}}^{\mathcal{M}} \in \mathfrak{D}(\rho, m+k) \subseteq \operatorname{span} \mathfrak{D}(\rho, m+k) = \operatorname{span} \mathfrak{D}(\rho, m)$. We have
        \begin{align*}
            \rho_{b|a, \mathcal{S}}^{\mathcal{M}}
            & = U_{a[\abs{a}]} \rho_{b|a^-, \mathcal{S}}^{\mathcal{M}} U_{a[\abs{a}]}^\dag \\
            & \in U_{a[\abs{a}]} \left( \operatorname{span} \mathfrak{D}(\rho, m) \right) U_{a[\abs{a}]}^\dag \\
            & = \operatorname{span} \left( U_{a[\abs{a}]} \mathfrak{D}(\rho, m) U_{a[\abs{a}]}^\dag \right) \\
            & \subseteq \operatorname{span} \mathfrak{D}(\rho, m+1) = \operatorname{span} \mathfrak{D}(\rho, m).
        \end{align*}

        Both cases together yield $\rho_{b|a, \mathcal{S}} \in \operatorname{span} \mathfrak{D}(\rho, m)$, and thus $\operatorname{span} \mathfrak{D}(\rho, m+k+1) \subseteq \operatorname{span} \mathfrak{D}(\rho, m)$. Because $\operatorname{span} \mathfrak{D}(\rho, m) \subseteq \operatorname{span} \mathfrak{D}(\rho, m+k+1)$, we conclude that $\operatorname{span} \mathfrak{D}(\rho, m+k+1) = \operatorname{span} \mathfrak{D}(\rho, m)$.

        {\vskip 4pt}

        \textbf{Conclusion}. $\operatorname{span} \mathfrak{D}(\rho, m) = \operatorname{span} \mathfrak{D}(\rho, m+k)$ for all $k \in \mathbb{N}$.
    \end{proof}

    Proposition \ref{prop-span} claims that $\dim \operatorname{span} \mathfrak{D}(\rho, m)$ either strictly increases (at least $1$) or reaches the maximum value. Note that $\dim \operatorname{span} \mathfrak{D}(\rho, 0) = 1$, we have:

    \begin{proposition} \label{col-span}
    $\operatorname{span} \mathfrak{D}(\rho, n^2-1) \supseteq \operatorname{span} \mathfrak{D}(\rho, m)$ for every $m \in \mathbb{N}$.
    \end{proposition}

    Now it is sufficient to prove the following:
    \begin{proposition}
        $\rho_s \sim \rho_t \Longleftrightarrow \rho_s \sim^{n^2-1} \rho_t$.
    \end{proposition}

    \begin{proof}
        ``$\Longrightarrow$'' Obvious.

        ``$\Longleftarrow$'' Suppose that $\rho_s \sim^{n^2-1} \rho_t$, then
        \[
            \Pr\nolimits_{\rho_s}^{\mathcal{M}}(b|a, \mathcal{S}) = \Pr\nolimits_{\rho_t}^{\mathcal{M}}(b|a, \mathcal{S})
        \]
        for every $a \in \Sigma^*$, $\mathcal{S} \in \mathfrak{S}_a$ and $b \in \Gamma^{\abs{\mathcal{S}}}$ with $\abs{a}+\abs{\mathcal{S}} \leq n^2 - 1$. That is,
        \[
            \operatorname{tr}(\rho_{b|a, \mathcal{S}}^{\mathcal{M}}) = 0
        \]
        where $\rho = \rho_s - \rho_t$.

        On the other hand, for every $a \in \Sigma^*$, $\mathcal{S} \in \mathfrak{S}_a$ and $b \in \Gamma^{\abs{\mathcal{S}}}$, by Proposition \ref{col-span}, we have $\rho_{b|a, \mathcal{S}}^{\mathcal{M}} \in \operatorname{span} \mathfrak{D}(\rho, n^2-1)$, and then
        \[
            \rho_{b|a, \mathcal{S}}^{\mathcal{M}} = \sum_{a', \mathcal{S}', b': \rho_{b'|a', \mathcal{S}'}^{\mathcal{M}} \in \mathfrak{D}(\rho, n^2-1)} \alpha_{b'|a', \mathcal{S}'} \rho_{b'|a', \mathcal{S}'}^{\mathcal{M}}
        \]
        for some coefficients $\alpha_{b'|a', \mathcal{S}'}$. Then
        \[
            \operatorname{tr}(\rho_{b|a, \mathcal{S}}^{\mathcal{M}}) = \sum_{a', \mathcal{S}', b': \rho_{b'|a', \mathcal{S}'}^{\mathcal{M}} \in \mathfrak{D}(\rho, n^2-1)} \alpha_{b'|a', \mathcal{S}'} \operatorname{tr}(\rho_{b'|a', \mathcal{S}'}^{\mathcal{M}}) = 0,
        \]
        i.e. $\Pr\nolimits_{\rho_s}^{\mathcal{M}}(b|a, \mathcal{S}) = \Pr\nolimits_{\rho_t}^{\mathcal{M}}(b|a, \mathcal{S})$, which immediately yields $\rho_s \sim \rho_t$.
    \end{proof}

    It is trivial that $\rho_s \sim \rho_t \Longrightarrow \rho_s \sim_{n^2-1} \rho_t \Longrightarrow \rho_s \sim^{n^2-1} \rho_t$. So, we complete the proof.

\section{Correctness of the algorithms} \label{app2}

In Sec. \ref{proofs}, we only analyse the complexities of Algorithms \ref{algo1} and \ref{algo2}. Here, we prove their correctness.

\subsection{Correctness of Algorithm \ref{algo1}}

    The correctness of the algorithm is proved in the following steps:

    {\vskip 4pt}

    \textbf{Step 1}. The algorithm always terminates.

    Note that $\mathcal{H}$ is a finite-dimensional Hilbert space. Let $n = \dim \mathcal{H} < \infty$. The algorithm guarantees that $\mathfrak{B}$ consists of linearly independent elements, whose dimension is bounded by $n^2$. Thus the number of times of modifications of $\mathfrak{B}$ is always bounded by $n^2$, or there must be two elements in $\mathfrak{B}$ that are linearly dependent. Only when $\mathfrak{B}$ is added a new element, the queue $Q$ will be pushed into some other (finite) elements. On the other hand, the algorithm pops one element from $Q$ in every iteration of the \textquotedblleft while\textquotedblright\ loop. Thus, $Q$ will become empty at some time and the algorithm terminates.

 {\vskip 4pt}

    \textbf{Step 2}. The queue $Q$ is monotonic.

    We define $\operatorname{ord}: \operatorname{dom}(\operatorname{ord}) \to \mathbb{N}$ to be the order of every valid tuple $(a, \mathcal{S}, b)$, where $$\operatorname{dom}(\operatorname{ord}) = \{ (a, \mathcal{S}, b) : a \in \Sigma^*, \mathcal{S} \in \mathfrak{S}_a, b \in \Gamma^{\abs{\mathcal{S}}} \} \subseteq \Sigma^* \times \mathfrak{S} \times \Gamma^*$$ is the defining domain of $\operatorname{ord}$. For convenience, let $\mathcal{S} = \{ s_1, s_2, \dots, s_{\abs{\mathcal{S}}} \}$ with $0 \leq s_1 \leq s_2 \leq \dots \leq s_{\abs{\mathcal{S}}} \leq \abs{a}$. Define the total order ``$<$'' on $\Sigma \cup \Gamma$ by
    \begin{enumerate}
      \item $\sigma_i < \sigma_j$ if $1 \leq i < j \leq \abs{\Sigma}$.
      \item $\gamma_i < \gamma_j$ if $1 \leq i < j \leq \abs{\Gamma}$.
      \item $\sigma_i < \gamma_j$ for every $1 \leq i \leq \abs{\Sigma}$ and $1 \leq j \leq \abs{\Gamma}$.
    \end{enumerate}
    Also define:
    \[
        (a, \mathcal{S}, b)^- = \begin{cases}
            (a, \mathcal{S}^-, b^-) & s_{\abs{\mathcal{S}}} = \abs{a}, \\
            (a^-, \mathcal{S}, b) & \text{otherwise},
        \end{cases}
    \]
    where $a^- = a[1:\abs{a}-1]$, $b^- = b[1:\abs{b}-1]$ and $\mathcal{S}^- = \{ s_1, s_2, \dots, s_{\abs{\mathcal{S}}-1} \}$, and
    \[
        \text{end}(a, \mathcal{S}, b) = \begin{cases}
            b[\abs{\mathcal{S}}] & s_{\abs{\mathcal{S}}} = \abs{a}, \\
            a[\abs{a}] & \text{otherwise},
        \end{cases}
    \]
    We further define $\operatorname{ord}(a, \mathcal{S}, b)$ recursively as follows:
    \begin{enumerate}
      \item $\operatorname{ord}(\cdot, \cdot, \cdot)$ is a bijection. That is, every tuple $(a, \mathcal{S}, b)$ corresponds to a unique number, and vice versa.
      \item $\operatorname{ord}(\epsilon, \emptyset, \epsilon) = 0$.
      \item For every two tuples $(a_1, \mathcal{S}_1, b_1)$ and $(a_2, \mathcal{S}_2, b_2)$, $\operatorname{ord}(a_1, \mathcal{S}_1, b_1) < \operatorname{ord}(a_2, \mathcal{S}_2, b_2)$ if and only if one of the following conditions holds:
          \begin{enumerate}
            \item $\abs{a_1}+\abs{\mathcal{S}_1} < \abs{a_2}+\abs{\mathcal{S}_2}$.
            \item $\abs{a_1}+\abs{\mathcal{S}_1} = \abs{a_2}+\abs{\mathcal{S}_2}$ and $\operatorname{ord}(a_1, \mathcal{S}_1, b_1)^- < \operatorname{ord}(a_2, \mathcal{S}_2, b_2)^-$.
            \item $\abs{a_1}+\abs{\mathcal{S}_1} = \abs{a_2}+\abs{\mathcal{S}_2}$, $\operatorname{ord}(a_1, \mathcal{S}_1, b_1)^- = \operatorname{ord}(a_2, \mathcal{S}_2, b_2)^-$ and $\text{end}(a_1, \mathcal{S}_1, b_1) < \text{end}(a_2, \mathcal{S}_2, b_2)$.
          \end{enumerate}
    \end{enumerate}
    Clearly, the queue $Q$ in the algorithm is monotonic in the increasing order of $\operatorname{ord}(a, \mathcal{S}, b)$.

     {\vskip 4pt}

    \textbf{Step 3}. $\operatorname{span} \mathfrak{B} = \operatorname{span} \mathfrak{D} (\rho, n^2-1)$.

    It is sufficient to verify the following:
    \begin{proposition}
        $\rho_{b|(a, \mathcal{S})} \in \operatorname{span} \mathfrak{B}$ for every $(a, \mathcal{S}, b) \in \operatorname{dom}(\operatorname{ord})$, where $\mathfrak{B}$ is the set $\mathfrak{B}$ in Algorithm \ref{algo1} after it terminates.
    \end{proposition}
    \begin{proof}
        Strengthen the proposition: $\rho_{b|(a, \mathcal{S})} \in \operatorname{span} \mathfrak{B}^{(\operatorname{ord}(a, \mathcal{S}, b))}$ for every $(a, \mathcal{S}, b) \in \operatorname{dom}(\operatorname{ord})$, where $$\mathfrak{B}^{(k)} = \{ (a, \mathcal{S}, b) \in \mathfrak{B}: \operatorname{ord}(a, \mathcal{S}, b) \leq k \}.$$
        We prove it by induction on $\operatorname{ord}(a, \mathcal{S}, b)$.

         {\vskip 4pt}

        \textbf{Basis}. $\operatorname{ord}(a, \mathcal{S}, b) = 0$, i.e. $(a, \mathcal{S}, b) = (\epsilon, \emptyset, \epsilon)$. Then $\rho_{\epsilon|(\epsilon, \emptyset)}$ is put into $\mathfrak{B}$ because $\mathfrak{B}$ is set to be $\emptyset$ initially. Thus $\rho_{\epsilon|(\epsilon, \emptyset)} \in \mathfrak{B}^{(0)} \subseteq \operatorname{span} \mathfrak{B}^{(0)}$.

         {\vskip 4pt}

        \textbf{Induction}. For every $(a, \mathcal{S}, b) \in \operatorname{dom}(\operatorname{ord})$ with $\operatorname{ord}(a, \mathcal{S}, b) \geq 1$, assume that
        \[
        \rho_{b'|(a', \mathcal{S}')} \in \operatorname{span} \mathfrak{B}^{(\operatorname{ord}(a', \mathcal{S}', b'))}
        \]
        for every $(a', \mathcal{S}', b') \in \operatorname{dom}(\operatorname{ord})$ with $\operatorname{ord}(a', \mathcal{S}', b') < \operatorname{ord}(a, \mathcal{S}, b)$.

         {\vskip 4pt}

        \textbf{Case 1}. $(a, \mathcal{S}, b)$ once appears in $Q$: Then the algorithm guarantees that $\rho_{b|(a, \mathcal{S})} \in \operatorname{span} \mathfrak{B}^{(\operatorname{ord}(a, \mathcal{S}, b))}$, because the algorithm checks whether $\rho_{b|(a, \mathcal{S})} \in \operatorname{span} \mathfrak{B}$ at that time, and if not, push $\rho_{b|(a, \mathcal{S})}$ into $\mathfrak{B}$.

         {\vskip 4pt}

        \textbf{Case 2}. $(a, \mathcal{S}, b)$ never appears in $Q$. Consider the following:

         {\vskip 4pt}

        \textbf{Subcase 2.1}. $\mathcal{S}$ is measure-closed, i.e. $s_{\abs{\mathcal{S}}} = \abs{a}$: Note that $\operatorname{ord}(a, \mathcal{S}^-, b^-) < \operatorname{ord}(a, \mathcal{S}, b)$, by the induction hypothesis, we have $\rho_{b^-|(a, \mathcal{S}^-)} \in \operatorname{span} \mathfrak{B}^{(\operatorname{ord}(a, \mathcal{S}^-, b^-))}$, then
        \begin{align*}
            \rho_{b|(a, \mathcal{S})}
            & = M_{b[\abs{\mathcal{S}}]} \rho_{b^-|(a, \mathcal{S}^-)} M_{b[\abs{\mathcal{S}}]}^\dag \\
            & \in M_{b[\abs{\mathcal{S}}]} \left( \operatorname{span} \mathfrak{B}^{(\operatorname{ord}(a, \mathcal{S}^-, b^-))} \right) M_{b[\abs{\mathcal{S}}]}^\dag \\
            & = \operatorname{span} \left( M_{b[\abs{\mathcal{S}}]} \mathfrak{B}^{(\operatorname{ord}(a, \mathcal{S}^-, b^-))} M_{b[\abs{\mathcal{S}}]}^\dag \right) \\
            & \subseteq \operatorname{span} \mathfrak{B}^{(\operatorname{ord}(a, \mathcal{S}^- + \{ \abs{a} \}, b^-b[\abs{\mathcal{S}}]))} = \operatorname{span} \mathfrak{B}^{(\operatorname{ord}(a, \mathcal{S}, b))}.
        \end{align*}

        \textbf{Subcase 2.2}. $\mathcal{S}$ is not measure-closed: Note that $\operatorname{ord}(a^-, \mathcal{S}, b) < \operatorname{ord}(a, \mathcal{S}, b)$, by the induction hypothesis, we have $\rho_{b|(a^-, \mathcal{S})} \in \operatorname{span} \mathfrak{B}^{(\operatorname{ord}(a^-, \mathcal{S}, b))}$, then
        \[
            \rho_{b|(a^-, \mathcal{S})} = \sum_{\rho_{b'|(a', \mathcal{S}')} \in \mathfrak{B}^{(\operatorname{ord}(a^-, \mathcal{S}, b))}} \alpha_{b'|(a', \mathcal{S}')} \rho_{b'|(a', \mathcal{S}')}
        \]
        for some coefficients $\alpha_{b'|(a', \mathcal{S}')}$. Thus,
        \begin{align*}
            \rho_{b|(a, \mathcal{S})}
            & = U_{a[\abs{a}]} \rho_{b|(a^-, \mathcal{S})} U_{a[\abs{a}]}^\dag \\
            & \in U_{a[\abs{a}]} \left( \operatorname{span} \mathfrak{B}^{(\operatorname{ord}(a^-, \mathcal{S}, b))} \right) U_{a[\abs{a}]}^\dag \\
            & = \operatorname{span} \left( U_{a[\abs{a}]} \mathfrak{B}^{(\operatorname{ord}(a^-, \mathcal{S}, b))} U_{a[\abs{a}]}^\dag \right) \\
            & \subseteq \operatorname{span} \mathfrak{B}^{(\operatorname{ord}(a^-a[\abs{a}], \mathcal{S}, b))} = \operatorname{span} \mathfrak{B}^{(\operatorname{ord}(a, \mathcal{S}, b))}.
        \end{align*}

        \textbf{Conclusion}. $\rho_{b|(a, \mathcal{S})} \in \operatorname{span} \mathfrak{B}^{(\operatorname{ord}(a, \mathcal{S}, b))}$ for every $(a, \mathcal{S}, b) \in \operatorname{dom}(\operatorname{ord})$.
    \end{proof}

    \textbf{Step 4}. $\rho_s \sim \rho_t$ if and only if $\operatorname{tr}(\varrho) = 0$ for every $\varrho \in \mathfrak{B}$, which is immediately obtained from Theorem \ref{thm-state-eq} Part 1.

\subsection{Correctness of Algorithm \ref{algo2}}

    The correctness of the algorithm is proved in the following steps:

     {\vskip 4pt}

    \textbf{Step 1}. The algorithm always terminates.

    Since $\mathcal{H}$ is a finite-dimensional Hilbert space, let $n = \dim \mathcal{H} < \infty$. The algorithm guarantees that $\mathfrak{B}_i (0 \leq i \leq k)$ consists of linearly independent elements, whose dimension is bounded by $n^2$. Thus for each $0 \leq i \leq k$, the number of times of modifications of $\mathfrak{B}_i$ is always bounded by $n^2$, or there must be two elements in $\mathfrak{B}_i$ that are linearly dependent. Only when $\mathfrak{B}_i$ is added a new element for some $0 \leq i \leq k$, will the queue $Q$ be pushed into some other (finite) elements. On the other hand, the algorithm pops one element from $Q$ in every iteration of the \textquotedblleft while\textquotedblright\  loop. Thus $Q$ will become empty at some time and the algorithm terminates.

 {\vskip 4pt}

    \textbf{Step 2}. The queue $Q$ is monotonic.

    Similar to the analysis of Algorithm \ref{algo1}, we define $\operatorname{ord}: \operatorname{dom}(\operatorname{ord}) \to \mathbb{N}$ be the order of every valid tuple $(a, \mathcal{S}, b)$, where $$\operatorname{dom}(\operatorname{ord}) = \{ (a, \mathcal{S}, b) : a \in \Sigma^*, \mathcal{S} \in \mathfrak{S}_a, b \in \Gamma^{\abs{\mathcal{S}}}, \abs{\mathcal{S}} \leq k \} \subseteq \Sigma^* \times \mathfrak{S} \times \Gamma^*$$ is the defining domain of $\operatorname{ord}$.
    Clearly, the queue $Q$ in the algorithm is monotonic in the increasing order of $\operatorname{ord}(a, \mathcal{S}, b)$.

     {\vskip 4pt}

    \textbf{Step 3}. $\operatorname{span} \mathfrak{B}_i \subseteq \operatorname{span} \mathfrak{B}_{i+1}$ for every $0 \leq i < k$.

    Obviously, this is guaranteed by the algorithm.

     {\vskip 4pt}

    \textbf{Step 4}. $\operatorname{span} \mathfrak{B}_i = \operatorname{span} \mathfrak{D}_i (\rho, n^2-1)$ for $0 \leq i \leq k$.

    It is sufficient to prove that following:
    \begin{proposition}
        $\rho_{b|(a, \mathcal{S})} \in \operatorname{span} \mathfrak{B}_{\abs{\mathcal{S}}}$ for every $(a, \mathcal{S}, b) \in \operatorname{dom}(\operatorname{ord})$, where $\mathfrak{B}_{\abs{\mathcal{S}}}$ is the set $\mathfrak{B}_{\abs{\mathcal{S}}}$ in Algorithm \ref{algo2} after it terminates.
    \end{proposition}
    \begin{proof}
        Strengthen the proposition: $\rho_{b|(a, \mathcal{S})} \in \operatorname{span} \mathfrak{B}^{(\operatorname{ord}(a, \mathcal{S}, b))}_{\abs{\mathcal{S}}}$ for every $(a, \mathcal{S}, b) \in \operatorname{dom}(\operatorname{ord})$, where $$\mathfrak{B}^{(k)}_i = \{ (a, \mathcal{S}, b) \in \mathfrak{B}: \operatorname{ord}(a, \mathcal{S}, b) \leq k, \abs{\mathcal{S}} \leq i \}.$$
        We prove it by induction on $\operatorname{ord}(a, \mathcal{S}, b)$.

         {\vskip 4pt}

        \textbf{Basis}. $\operatorname{ord}(a, \mathcal{S}, b) = 0$, i.e. $(a, \mathcal{S}, b) = (\epsilon, \emptyset, \epsilon)$. Then $\rho_{\epsilon|(\epsilon, \emptyset)}$ is put into $\mathfrak{B}_i(0 \leq i \leq k)$ because $\mathfrak{B}_i(0 \leq i \leq k)$ is set to be $\emptyset$ initially. Thus $\rho_{\epsilon|(\epsilon, \emptyset)} \in \mathfrak{B}^{(0)}_i \subseteq \operatorname{span} \mathfrak{B}^{(0)}_i$ for $0 \leq i \leq k$.

         {\vskip 4pt}

        \textbf{Induction}. For every $(a, \mathcal{S}, b) \in \operatorname{dom}(\operatorname{ord})$ with $\operatorname{ord}(a, \mathcal{S}, b) \geq 1$, assume that
        \[
        \rho_{b'|(a', \mathcal{S}')} \in \operatorname{span} \mathfrak{B}^{(\operatorname{ord}(a', \mathcal{S}', b'))}_{\abs{\mathcal{S}}}
        \]
        for every $(a', \mathcal{S}', b') \in \operatorname{dom}(\operatorname{ord})$ with $\operatorname{ord}(a', \mathcal{S}', b') < \operatorname{ord}(a, \mathcal{S}, b)$.

         {\vskip 4pt}

        \textbf{Case 1}. $(a, \mathcal{S}, b)$ once appears in $Q$: Then the algorithm guarantees that $\rho_{b|(a, \mathcal{S})} \in \operatorname{span} \mathfrak{B}^{(\operatorname{ord}(a, \mathcal{S}, b))}_{\abs{\mathcal{S}}}$, because the algorithm checks whether $\rho_{b|(a, \mathcal{S})} \in \operatorname{span} \mathfrak{B}_{\abs{\mathcal{S}}}$ at that time, and if not, push $\rho_{b|(a, \mathcal{S})}$ into $\mathfrak{B}_{\abs{\mathcal{S}}}$.

         {\vskip 4pt}

        \textbf{Case 2}. $(a, \mathcal{S}, b)$ never appears in $Q$. Consider the following subcases:

         {\vskip 4pt}

        \textbf{Subcase 2.1}. $\mathcal{S}$ is measure-closed, i.e. $s_{\abs{\mathcal{S}}} = \abs{a}$: Note that $\operatorname{ord}(a, \mathcal{S}^-, b^-) < \operatorname{ord}(a, \mathcal{S}, b)$, by the induction hypothesis, we have $\rho_{b^-|(a, \mathcal{S}^-)} \in \operatorname{span} \mathfrak{B}^{(\operatorname{ord}(a, \mathcal{S}^-, b^-))}_{\abs{\mathcal{S}}-1}$, then
        \begin{align*}
            \rho_{b|(a, \mathcal{S})}
            & = M_{b[\abs{\mathcal{S}}]} \rho_{b^-|(a, \mathcal{S}^-)} M_{b[\abs{\mathcal{S}}]}^\dag \\
            & \in M_{b[\abs{\mathcal{S}}]} \left( \operatorname{span} \mathfrak{B}^{(\operatorname{ord}(a, \mathcal{S}^-, b^-))}_{\abs{\mathcal{S}}-1} \right) M_{b[\abs{\mathcal{S}}]}^\dag \\
            & = \operatorname{span} \left( M_{b[\abs{\mathcal{S}}]} \mathfrak{B}^{(\operatorname{ord}(a, \mathcal{S}^-, b^-))}_{\abs{\mathcal{S}}-1} M_{b[\abs{\mathcal{S}}]}^\dag \right) \\
            & \subseteq \operatorname{span} \mathfrak{B}^{(\operatorname{ord}(a, \mathcal{S}^- + \{ \abs{a} \}, b^-b[\abs{\mathcal{S}}]))}_{\abs{\mathcal{S}}} = \operatorname{span} \mathfrak{B}^{(\operatorname{ord}(a, \mathcal{S}, b))}_{\abs{\mathcal{S}}}.
        \end{align*}

        \textbf{Subcase 2.2}. $\mathcal{S}$ is not measure-closed: Note that $\operatorname{ord}(a^-, \mathcal{S}, b) < \operatorname{ord}(a, \mathcal{S}, b)$, by the induction hypothesis, we have $\rho_{b|(a^-, \mathcal{S})} \in \operatorname{span} \mathfrak{B}^{(\operatorname{ord}(a^-, \mathcal{S}, b))}_{\abs{\mathcal{S}}}$, then
        \begin{align*}
            \rho_{b|(a, \mathcal{S})}
            & = U_{a[\abs{a}]} \rho_{b|(a^-, \mathcal{S})} U_{a[\abs{a}]}^\dag \\
            & \in U_{a[\abs{a}]} \left( \operatorname{span} \mathfrak{B}^{(\operatorname{ord}(a^-, \mathcal{S}, b))}_{\abs{\mathcal{S}}} \right) U_{a[\abs{a}]}^\dag \\
            & = \operatorname{span} \left( U_{a[\abs{a}]} \mathfrak{B}^{(\operatorname{ord}(a^-, \mathcal{S}, b))}_{\abs{\mathcal{S}}} U_{a[\abs{a}]}^\dag \right) \\
            & \subseteq \operatorname{span} \mathfrak{B}^{(\operatorname{ord}(a^-a[\abs{a}], \mathcal{S}, b))}_{\abs{\mathcal{S}}} = \operatorname{span} \mathfrak{B}^{(\operatorname{ord}(a, \mathcal{S}, b))}_{\abs{\mathcal{S}}}.
        \end{align*}

        \textbf{Conclusion}. $\rho_{b|(a, \mathcal{S})} \in \operatorname{span} \mathfrak{B}^{(\operatorname{ord}(a, \mathcal{S}, b))}_{\abs{\mathcal{S}}}$ for every $(a, \mathcal{S}, b) \in \operatorname{dom}(\operatorname{ord})$.
    \end{proof}

    \textbf{Step 4}. $\rho_s \sim_{k} \rho_t$ if and only if $\operatorname{tr}(\varrho) = 0$ for every $\varrho \in \mathfrak{B}_i$ for $0 \leq i \leq k$, which is immediately obtained from Theorem \ref{thm-state-eq} Part 2.

\section{An efficient description of Problem \ref{prob1}} \label{app3}

A simple form of Problem \ref{prob1}, analog to Proposition \ref{prop-eq-k} about Problem \ref{prob2}, is given in the following:

\begin{proposition} \label{prop-eq}
        Let $\mathcal{M}_1 = (\Sigma, \Gamma, \mathcal{H}^{(1)}, U^{(1)}, M^{(1)})$ and $\mathcal{M}_2 = (\Sigma, \Gamma, \mathcal{H}^{(2)}, U^{(2)}, M^{(2)})$ be two QMMs with initial states $\rho_1$ and $\rho_2$, respectively. Let $n_1 = \dim \mathcal{H}_1$ and $n_2 = \dim \mathcal{H}_2$. Then
        $(\mathcal{M}_1, \rho_1) \sim (\mathcal{M}_2, \rho_2)$ if and only if there is a $(n_1^2+n_2^2) \times (n_1^2+n_2^2)$ matrix $M_c$ for every $c \in \Sigma \cup \Gamma$ and a $(n_1^2+n_2^2) \times (n_1^2+n_2^2)$ matrix $F$ such that
        \begin{enumerate}
          \item $F_{\cdot, 1} = \begin{bmatrix} \vec \rho_1 \\ \vec \rho_2 \end{bmatrix}$.
          \item $\eta^\dag F = 0$, where $\eta = \begin{bmatrix} \eta_1 \\ -\eta_2 \end{bmatrix}$, $\eta_1$ and $\eta_2$ are the vectorizations of trace for $\mathcal{M}_1$ and $\mathcal{M}_2$, respectively.
          \item For $c \in \Sigma$,
          \[
            \begin{bmatrix}
                \hat U_c^{(1)} & 0 \\
                0 & \hat U_c^{(2)}
            \end{bmatrix} F = F M_c.
          \]
          \item For $c \in \Gamma$,
          \[
            \begin{bmatrix}
                \hat M_c^{(1)} & 0 \\
                0 & \hat M_c^{(2)}
            \end{bmatrix} F = F M_c.
          \]
        \end{enumerate}
    \end{proposition}

    The conditions of Proposition \ref{prop-eq} on $\mathcal{M}_2$, including that it be a QMM, can be phrased in $O((\abs{\Sigma}+\abs{\Gamma})(n_1+n_2)^4)$ polynomials of degree $d = 3$ in $O((\abs{\Sigma}+\abs{\Gamma})(n_1+n_2)^4)$ variables, better than $O((\abs{\Sigma}+\abs{\Gamma})(n_1+n_2)^6)$ polynomials of degree $d = 3$ in $O((\abs{\Sigma}+\abs{\Gamma})(n_1+n_2)^6)$ variables given by Proposition \ref{prop-eq-k} when $k = O((n_1+n_2)^2)$. By Theorem \ref{thm-existence}, Problem \ref{prob1} is solvable in \textbf{PSPACE}.
    % Moreover, exact QMM $\mathcal{M}$ and state $\rho$ can be computed in \textbf{EXPSPACE}.

\begin{proof}
``$\Longrightarrow$'' If $(\mathcal{M}_1, \rho_1) \sim (\mathcal{M}_2, \rho_2)$, then
\[
    \operatorname{tr}((\rho_1)^{\mathcal{M}_1}_{b|a,\mathcal{S}}) = \operatorname{tr}((\rho_2)^{\mathcal{M}_2}_{b|a,\mathcal{S}})
\]
for every $a \in \Sigma^*$, $\mathcal{S} \in \mathfrak{S}$ and $b \in \Gamma^{\abs{\mathcal{S}}}$. Let $\rho = \rho_1 \oplus \rho_2$, then
\[
    \rho^{\mathcal{M}_1 \oplus \mathcal{M}_2}_{b|a,\mathcal{S}} = (\rho_1)^{\mathcal{M}_1}_{b|a,\mathcal{S}} \oplus (\rho_2)^{\mathcal{M}_2}_{b|a,\mathcal{S}}.
\]
Let $n = n_1^2 + n_2^2$,
let $\rho^{(1)}, \rho^{(2)}, \dots, \rho^{(n)} \in \mathfrak{D}(\rho, n-1)$ be the basis of $\operatorname{span} \mathfrak{D}(\rho, n-1)$ with $\rho^{(1)} = \rho$, and let $\rho^{(i)} = \rho^{(i)}_1 \oplus \rho^{(i)}_2$ and $$f_i = \begin{bmatrix} \vec \rho^{(i)}_1 \\ \vec \rho^{(i)}_2 \end{bmatrix}.$$ Note that $\operatorname{tr}(\rho^{(i)}_1) = \operatorname{tr}(\rho^{(i)}_2)$. Define $ F = \begin{bmatrix}
        f_1 & f_2 & \dots & f_n
    \end{bmatrix}. $
Note that $$\eta^\dag f_i = \eta_1 \vec \rho^{(i)}_1 - \eta_2 \vec \rho^{(i)}_2 = \operatorname{tr}(\rho^{(i)}_1) - \operatorname{tr}(\rho^{(i)}_2) = 0.$$ We conclude that $\eta^\dag F = 0$.

For every $\sigma \in \Sigma$, for every $\rho^{(j)} = \rho_{b|a,\mathcal{S}} \in \mathfrak{D}(\rho, n-1)$,
by Proposition \ref{col-span}, $\rho_{b|a \sigma, \mathcal{S}} \in \operatorname{span} \mathfrak{D}(\rho, n-1)$, then
\[
    \rho_{b|a \sigma, \mathcal{S}} = \sum_{i=1}^n \alpha_{ij} \rho^{(i)}
\]
for some coefficients $\alpha_{ij}$, i.e.
\begin{align*}
    & (\rho_1)_{b|a \sigma, \mathcal{S}} = U^{(1)}_\sigma \rho^{(j)}_1 (U^{(1)}_\sigma)^\dag = \sum_{i=1}^n \alpha_{ij} \rho^{(i)}_1, \\
    & (\rho_2)_{b|a \sigma, \mathcal{S}} = U^{(2)}_\sigma \rho^{(j)}_2 (U^{(2)}_\sigma)^\dag = \sum_{i=1}^n \alpha_{ij} \rho^{(i)}_2.
\end{align*}
Then
\begin{align*}
    & \hat U^{(1)}_\sigma \vec \rho^{(j)}_1 = \sum_{i=1}^n \alpha_{ij} \vec \rho^{(i)}_1, \\
    & \hat U^{(2)}_\sigma \vec \rho^{(j)}_2 = \sum_{i=1}^n \alpha_{ij} \vec \rho^{(i)}_2.
\end{align*}
That is,
\[
    \begin{bmatrix}
        \hat U^{(1)}_\sigma & 0 \\
        0 & \hat U^{(2)}_\sigma
    \end{bmatrix} f_j = F \begin{bmatrix}
        \alpha_{1j} \\
        \alpha_{2j} \\
        \vdots \\
        \alpha_{nj}
    \end{bmatrix},
\]
and we obtain that
\[
    \begin{bmatrix}
        \hat U^{(1)}_\sigma & 0 \\
        0 & \hat U^{(2)}_\sigma
    \end{bmatrix} F = F M_\sigma,
\]
where $M_\sigma = [\alpha_{ij}]$.

For every $m \in \Gamma$, similarly, we have
\[
    \begin{bmatrix}
        \hat M^{(1)}_m & 0 \\
        0 & \hat M^{(2)}_m
    \end{bmatrix} F = F M_m
\]
for some $M_m$.

{\vskip 4pt}

``$\Longleftarrow$''. For every $a \in \Sigma^*$, $\mathcal{S} \in \mathfrak{S}_a$ and $b \in \Gamma^{\abs{\mathcal{S}}}$,
\[
    \begin{bmatrix}
        \hat V_{b|a,\mathcal{S}}^{(1)} & 0 \\
        0 & \hat V_{b|a,\mathcal{S}}^{(2)}
    \end{bmatrix} F = F M_{b|a,\mathcal{S}},
\]
where $$M_{b|a,\mathcal{S}} = M_{a_1} M_{b_1} M_{a_2} M_{b_2} \dots M_{a_{\abs{\mathcal{S}}}} M_{b_{\abs{\mathcal{S}}}} M_{a_{\abs{\mathcal{S}}+1}}.$$
Then
\begin{align*}
    \operatorname{tr}((\rho_1)_{b|a,\mathcal{S}}) - \operatorname{tr}((\rho_2)_{b|a,\mathcal{S}})
    & = \eta_1^\dag (\vec \rho_1)_{b|a,\mathcal{S}} - \eta_2^\dag (\vec \rho_2)_{b|a,\mathcal{S}} \\
    & = \begin{bmatrix}
        \eta_1 \\
        -\eta_2
    \end{bmatrix}^\dag \begin{bmatrix}
        \hat V_{b|a,\mathcal{S}}^{(1)} & 0 \\
        0 & \hat V_{b|a,\mathcal{S}}^{(2)}
    \end{bmatrix} \begin{bmatrix}
        \vec \rho_1 \\
        \vec \rho_2
    \end{bmatrix} \\
    & = \eta^\dag \begin{bmatrix}
        \hat V_{b|a,\mathcal{S}}^{(1)} & 0 \\
        0 & \hat V_{b|a,\mathcal{S}}^{(2)}
    \end{bmatrix} F e_1 \\
    & = \eta^\dag F M_{b|a,\mathcal{S}} e_1 \\
    & = 0,
\end{align*}
where $e_1 = (1, 0, 0, \dots, 0)^T$.
\end{proof}

\end{document}